\documentclass[12pt]{article}
\pdfoutput=1
\usepackage{natbib}
\usepackage{multirow}
\usepackage{amsthm}
\usepackage{amssymb}
\usepackage{xcolor}
\usepackage{bbm}
\usepackage{lscape}
\usepackage{caption}
\usepackage{subcaption}
\usepackage{mathtools}
\usepackage{url}
\usepackage{float}
\usepackage{amsmath}
\usepackage{mathalfa}

\usepackage{algorithm}
\usepackage[parfill]{parskip}
\usepackage{algorithmic}
\usepackage[T1]{fontenc}
\usepackage{titling}
\usepackage[toc,page]{appendix}
\usepackage{listings}
\usepackage{geometry}
\usepackage{dsfont}
\usepackage{wrapfig}
\usepackage{graphicx}
\usepackage[english]{babel}
\newtheorem{theorem}{Theorem}[section]

\numberwithin{equation}{section}

\newcommand{\beginsupplement}{%
        \setcounter{table}{0}
        \renewcommand{\thetable}{S\arabic{table}}%
        \setcounter{figure}{0}
        \renewcommand{\thefigure}{S\arabic{figure}}%
        \setcounter{section}{0}
        \renewcommand{\thesection}{\arabic{section}}%
     }

\definecolor{ao}{rgb}{0.0, 0.5, 0.0}

\title{A preplanned multi-stage platform trial for discovering multiple superior treatments with control of FWER and power}
\author{Peter Greenstreet$^{1,\star}$, Thomas Jaki$^{2,3}$, Alun Bedding$^{4}$, Pavel Mozgunov$^{2}$
\\  \footnotesize{$^1$ Department of Mathematics and
Statistics, Lancaster University, Lancaster, UK} 
\\  \footnotesize{$^2$ MRC Biostatistics Unit, University of
Cambridge, Cambridge, UK} 
\\  \footnotesize{$^3$ University of Regensburg, Regensburg, Germany} 
\\  \footnotesize{$^4$ Roche Products Ltd., Welwyn Garden City, UK}
\\  \footnotesize{$^\star$ p.greenstreet@lancaster.ac.uk} }
\date{}
\begin{document}
\maketitle

\begin{abstract}
There is a growing interest in the implementation of platform trials, which provide the flexibility to incorporate new treatment arms during the trial and the ability to halt treatments early based on lack of benefit or observed superiority. In such trials, it can be important to ensure that error rates are controlled. This paper introduces a multi-stage design that enables the addition of new treatment arms, at any point, in a pre-planned manner within a platform trial, while still maintaining control over the family-wise error rate. This paper focuses on finding the required sample size to achieve a desired level of statistical power when treatments are continued to be tested even after a superior treatment has already been found. This may be of interest if there are other sponsors treatments which are also superior to the current control or multiple doses being tested. The calculations to determine 
the expected sample size is given. A motivating trial is presented in which the sample size of different configurations is studied. Additionally the approach is compared to running multiple separate trials and it is shown that in many scenarios if family wise error rate control is needed there may not be benefit in using a platform trial when comparing the sample size of the trial.
\end{abstract}

%




%
\section{Introduction}
\label{sec:Intro}
Platform trials are a type of trial design which can aim to reduce the amount of time and cost of clinical trials and
in recent years, there has been an increase in the utilization of such trials, including during the COVID-19 pandemic \citep{StallardNigel2020EADf,lee2021}. Clinical trials take many years to run and can cost billions of dollars \citep{MullardAsher2018Hmdp}. During this time it is not uncommon for new promising treatments to emerge and become ready to join the current phase later  \citep{Choodari-OskooeiBabak2020Anea}. Therefore it may be advantageous to include these treatments into an ongoing trial. This can have multiple potential benefits including: shared trial infrastructure; the possibility to use a shared control group; less administrative and logistical effort than setting up separate trials and enhance the recruitment \citep{BurnettThomas2020Aeta, MeurerWilliamJ2012ACTA}. This results in useful therapies potentially being identified faster while reducing cost and time \citep{cohen2015adding}. 
\par
There is an ongoing discussion about how to add new treatments to clinical trials \citep{cohen2015adding,lee2021} in both a pre-planned and in an unplanned manor \citep{GreenstreetPeter2021Ammp, BurnettThomas2020Aeta}. In both \cite{MaxineBennett2020Dfaa, Choodari-OskooeiBabak2020Anea} approaches are proposed which extend the Dunnett test \citep{DunnettCharlesW1955AMCP} to allow for unplanned  additional arms to be included into multi-arm trials while still controlling the family wise error rate (FWER). This methodology does not incorporate the possibility of interim analyses. 
\par
FWER is often considered to be one of the strongest types of type I error control in a multi-arm trial \citep{WasonJames2016Srfm}. There are other approaches one may wish to consider such as pairwise error rate (PWER) and the false discovery rate (FDR) \citep{RobertsonDavidS2022Oecf, CuiXinping2023SPfP, BrattonDanielJ2016TIer, Choodari-OskooeiBabak2020Anea}. However as discussed in \cite{WasonJamesMS2014Cfmi} there are scenarios where FWER is seen as the recommended error control, and it can be a regulatory requirement.
\par
One may wish to include interim analyses as they allow for ineffective treatments to be dropped for futility earlier and allow treatments to stop early if they are found superior to the control. Therefore potentially improving the efficiency of design of a clinical trial by decreasing the expected sample sizes and costs of a trial  \citep{PocockStuartJ.1977GSMi,ToddSusan2001Iaas,WasonJames2016Srfm}. Multi-arm multi-stage (MAMS) designs  \citep{MagirrD.2012AgDt, RoystonPatrick2003Ndfm} allow interim analyses while still allowing several treatments to be evaluated within one study, but do not allow for additional arms to be added throughout the trial. \cite{BurnettThomas2020Aeta} have developed an approach that builds on \cite{HommelGerhard2001AMoH} to incorporate unplanned additional treatment arms to be added to a trial already in progress using the conditional error principle \citep{ProschanMA1995DEoS}. This allows for modifications during the course of a trial. However due to the unplanned nature of the adaptation, later treatments can be greatly underpowered compared to arms which begin the trial.
\par
In a recent paper \cite{GreenstreetPeter2021Ammp} proposed a preplanned approach to adding additional arms in which interim analyses can be conducted and multiple arms can be evaluated with some arms being added at later time points. In this work the trial was powered assuming that only one treatment may be taken forward. However as discussed in the work by \cite{UrachS.2016Mgsd, SerraAlessandra2022Aorm} this may not always be the case. For example one may be interested in lower doses; or multiple treatments from different sponsors; or interested if another treatment has preferable secondary outcomes if it also meets the primary outcome. Furthermore in \cite{GreenstreetPeter2021Ammp} treatment arms can only be added when a interim analysis happens, this can greatly restrict when arms can join the trial resulting in potentially large time periods that a new treatment is available before an interim is conducted so able to join the trial.
\par
In this work, we provide an analytical method for adding of treatments at any point to a multi-arm multi-stage  trial in a pre-planned manner, while still controlling the statistical errors. This work will focus on trials in which one is interested in continuing to investigate the other treatments even after a superior treatment has been found. In addition in this work multiple types of power will be considered, and will prove that the conjunctive power of the study is at its lowest for a given sample size when all the active treatments have a clinically relevant effect, where the conjunctive power is the probability of finding all the active treatments with a clinically relevant effect. 
The methodology discussed here can be used to create multiple designs for each point the additional treatments may be added into the trial. 
 This is due to the model flexibility, as the additional arms do not need to be added when a interim happens, resulting in new active arms being able to be added faster into the platform trial.
\par
This work will focus predominantly on the case where one has equal allocation ratio across all the active treatments and the same number of interim analyses per treatment with the same boundary shape. This is to help mitigate issues with time trends  \citep{AltmanDouglasG.1988Theo, GetzKennethA2017TwTi}  when changing allocation ratio mid trial \citep{ProschanMichael2020RtTo, RoigMartaBofill2023Oasi}. However the proposed methodology is general and therefore can be implemented for when there is not equal allocation ratio across all the active treatments, however one needs to be cautious of potential time trend effects.
\par
We begin this work by analytically calculating the FWER and power of the study and use these to calculate both the stopping boundaries and sample size. Then in Section \ref{sec:expsam} the equations for sample size distribution and expected sample size are given. A trial example of FLAIR \citep{HowardDenaR2021Apti}, in Section \ref{sec:Trialexample}, is used to motivate a hypothetical trial of interest. The sample size and stopping boundaries are found for multiple types of power control and the effect of different treatment effects is studied. Then the trial designs are then compared to running multiple separate trials. Finally in Section \ref{sec:Discussion} there is a discussion of the paper and this introduces areas for further research.

\section{Methodology}
\label{sec:Setting}
\subsection{Setting}
\label{subsec:Setting}
In the clinical trial design considered in this work K experimental arms effectiveness is compared to a common control arm. The trial has $K^\star$  treatments starting at the beginning of the trial, and the remaining $K-K^\star$ treatments being added at later points into the platform. The primary outcome measured for each patient is assumed to be independent, continuous, and follows a normal distribution with a known variance ($\sigma^2$). 
\par
The points at which each active treatment arm is added are predetermined, but can be set to any point within the trial.
Each of the $K$ treatments is potentially tested at a series of analyses indexed by $j=1,\hdots,J_k$ where $J_k$ is the maximum number of analyses for a given treatment $k = 1,\hdots, K$. Let $n(k)$ denote the number of patients recruited to the control treatment before treatment $k$ is added to the platform trial and define the vector of adding times by $\mathbf{n(K)}=(n(1),\hdots,n(K))$. Therefore for treatments that start at the beginning of the trial $n(k)=0$. 
We also denote $n_{k,j}$ as the number of patients recruited to treatment $k$ by the end of it's $j^\text{th}$ stage and define $n_{0,k,j}$ as the total number of patients recruited to the control at the end of treatment $k$'s $j^\text{th}$ stage. We define $n_k=n_{k,1}$ as the number recruited to the first stage of treatment $k$, $k=1,\hdots,K$. Similarly we define $r_{k,j}$ and $r_{0,k,j}$ as the ratio of patients recruited treatment $k$ and the control by treatment $k$'s $j^\text{th}$ stage, respectively. Also $r(k)$ denotes the ratio of patients recruited to the control before treatment $k$  is added to the trial. 
For example if a trial was planned to have equal number of patients per stage and a treatment ($k'$) was added at the first interim then $r(k')=1$ and at the first stage for $k'$, $r_{0,k',1}=2$. The total sample size of a trial is denoted by $N$. The maximum total planned sample size is $\max(N)= \sum_{k=1}^K n_{k,J_k}+\max_{k\in 1,\hdots, K}(n_{0,k,J_k})$. 
\par
Throughout the trial, the control arm is recruited and maintained for the entire duration. The comparisons between the control arm and the active treatment arms are based on concurrent controls, meaning that only participants recruited to the control arm at the same time as the corresponding active arm are used in the comparisons. Work on the use of non-concurrent controls include \cite{lee2020including, MarschnerIanC2022Aoap}.
\par
The null hypotheses of interest are
$
H_{01}: \mu_1 \leq \mu_0, H_{02}: \mu_2 \leq \mu_0, ... , H_{0K}: \mu_K \leq \mu_0,
$
where $\mu_1, \hdots, \mu_K$ are the mean responses on the $K$ experimental treatments and $\mu_0$ is the mean response of the control group. The global null hypothesis, $\mu_0=\mu_1=\mu_2=\hdots=\mu_K$ is denoted by $H_{G}$. At analysis $j$ for treatment $k$, to test $H_{0k}$ it is assumed that responses, $X_{k,i}$, from patients $i=1,\hdots, n_{k,j}$ are observed, as well as the responses $X_{0,i}$ from patients $i=n(k)+1,\hdots, n_{0,k,j}$. These are the outcomes of the patients allocated to the control which have been recruited since treatment $k$ has been added into the trial up to the $j^\text{th}$ analysis of treatment $k$. The null hypotheses are tested using the test statistics
\begin{equation*}
Z_{k,j}=\frac{ n_{k,j}^{-1} \sum^{n_{k,j}}_{i=1} X_{k,i} - (n_{0,k,j}-n(k))^{-1} \sum^{n_{0,k,j}}_{i=n(k)+1} X_{0,i} }{\sigma\sqrt{(n_{k,j})^{-1}+ (n_{0,k,j}-n(k))^{-1} }}.
\end{equation*}

The decision-making for the trial is made by the upper and lower stopping boundaries, denoted as $U_{k}=(u_{k,1},\hdots,u_{k,J_k})$ and $L_{k}=(l_{k,1},\hdots,l_{k,J_k})$. These boundaries are utilized to determine whether to continue or halt a treatment arm or even the whole trial at various stages.  The decision-making process is as follows: if the test statistic for treatment $k$ at stage $j$ exceeds the upper boundary $u_{k,j}$, the null hypothesis $H_{0k}$ is rejected, and the treatment is stopped with the conclusion that it is superior to the control. Conversely, if $Z_{k,j}$ falls below the lower boundary $l_{k,j}$, treatment k is stopped for futility for all subsequent stages of the trial. If neither the superiority nor futility conditions are met, $l_{k,j} \leq Z_{k,j} \leq u_{k,j}$, treatment $k$ proceeds to it's next stage $j+1$. If all the active treatments are stopped than the trial stops. These bounds are found to control the type I error of desire for the trial. In this work we consider the familywise error rate (FWER) as the  type I error control of focus as discussed in Section \ref{Sec:FWER}.

\subsection{Familywise error rate (FWER)}
\label{Sec:FWER}
The FWER in the strong sense is a way of defining the type I error of a trial with multiple hypotheses and is defined as
\begin{equation}
P(\text{reject at least one true } H_{0k} \text{ under any null configuation}, k=1,\hdots, K) \leq \alpha.
\label{def:FWER}
\end{equation}
where $\alpha$ is the desired level of control for the FWER. 
As proven in \cite{GreenstreetPeter2021Ammp} which builds on \cite{MagirrD.2012AgDt}, one can show that the FWER is controlled in the strong sense under the global null, as given in the Supporting Information Section 1.
The FWER under the global null hypothesis is equal to

\begin{equation}
1-\sum_{\substack{j_k=1 \\ k=1,2,\ldots,K}}^{J_k} \Phi(\mathbf{L}_{\mathbf{j_k}}\mathbf{(0)},\mathbf{U}_{\mathbf{j_k}}\mathbf{(0)},\Sigma_{\mathbf{j_k}}).
\label{FWERequationnew}
\end{equation}
Here $\Phi(\cdot)$ denotes the multivariate standard normal distribution function, and $\mathbf{j_k}=(j_1,\hdots,j_K)$. With $\mathbf{j_k}$ one can define the vector of upper and lower limits from the multivariate standard normal distribution function as $\mathbf{U}_{\mathbf{j_k}}\mathbf{(0)}=(U_{1,j_1}(0), \hdots, U_{K,j_K}(0))$ and $\mathbf{L}_{\mathbf{j_k}}\mathbf{(0)}=(L_{1,j_1}(0), \hdots, L_{K,j_K}(0))$ where $U_{k,j_k}(0)=(u_{k,1},\hdots,l_{k,j_k})$ and $L_{k,j_k}(0)=(l_{k,1},\hdots,-\infty)$ respectively. $U_{k,j_k}(0)$ and $L_{k,j_k}(0)$ represent the upper and lower limits for treatment $k$ given $j_k$ for the multivariate standard normal distribution function.   The correlation matrix  $\Sigma_{\mathbf{j_k}}$ complete correlation structure is
\begin{equation}
\Sigma_{\mathbf{j_k}}= \begin{pmatrix}
\rho_{(1,1),(1,1)} & \rho_{(1,1),(1,2)} & \hdots & \rho_{(1,1),(1,j_1)} & \rho_{(1,1),(2,1)} & \hdots & \rho_{(1,1),(K,j_k)} \\
\\
\rho_{(1,2),(1,1)} & \rho_{(1,2),(1,2)} & \hdots & \rho_{(1,2),(1,j_1)} & \rho_{(1,2),(2,1)} & \hdots & \rho_{(1,2),(K,j_k)} 
\\
\vdots & \vdots & \ddots & \vdots & \vdots & \ddots & \vdots
\\
\rho_{(1,j_1),(1,1)} & \rho_{(1,j_1),(1,2)} & \hdots & \rho_{(1,j_1),(1,j_1)} & \rho_{(1,j_1),(2,1)} & \hdots & \rho_{(1,j_1),(K,j_k)} 
\\
\rho_{(2,1),(1,1)} & \rho_{(2,1),(1,2)} & \hdots & \rho_{(2,1),(1,j_1)} & \rho_{(2,1),(2,1)} & \hdots & \rho_{(2,1),(K,j_k)} 
\\
\vdots & \vdots & \ddots & \vdots & \vdots & \ddots & \vdots
\\
\rho_{(K,j_k),(1,1)} & \rho_{(K,j_k),(1,2)} & \hdots & \rho_{(K,j_k),(1,j_1)} & \rho_{(K,j_k),(2,1)} & \hdots & \rho_{(K,j_k),(K,j_k)} \\
\end{pmatrix}.
\label{equ:FullCor}
\end{equation}
where $\rho_{(k,j),(k^\star,j^\star)}$ equals one of the following: If $k=k^\star$ and $j=j^\star$ then $\rho_{(k,j),(k^\star,j^\star)}=1$; If $k=k^\star$ and $j<j^\star$ then
\begin{align*}
\rho_{(k,j),(k^\star,j^\star)}= \Bigg{(}\sqrt{r_{k,j}^{-1}+(r_{0,k,j}-r(k))^{-1}} & \sqrt{r_{k,j^\star}^{-1}+(r_{0,k,j^\star}-r(k))^{-1}} \Bigg{)}^{-1} 
\\
& \Bigg{(} \frac{1}{r_{k,j^\star}} + \frac{1}{r_{0,k,j^\star}-r(k)}  \Bigg{)};
\end{align*} 
and if $k \neq k^\star$ where $n(k)<n(k^\star)$ then
\begin{align*}
\rho_{(k,j),(k^\star,j^\star)}= \max\Bigg{[}0, &  \Bigg{(}\sqrt{r_{k,j}^{-1}+(r_{0,k,j}-r(k))^{-1}}\sqrt{r_{k^\star,j^\star}^{-1}+(r_{0,k^\star,j^\star}-r(k^\star))^{-1}} \Bigg{)}^{-1}
\\  & \Bigg{(} \frac{\min[r_{0,k,j}-r(k^\star) , r_{0,k^\star,j^\star}-r(k^\star))}{[r_{0,k,j}-r(k)] [r_{0,k^\star,j^\star}-r(k^\star)]}  \Bigg{)} \Bigg{]}.
\end{align*}
\par 
As seen here if treatment $k^\star$ is added to the platform trial after the $j$ stage for treatment $k$ then the correlation equals 0 as there is no shared controls.
The proposed methodology allows for different critical boundaries to be used for each treatment $k$ as shown in Equation \eqref{FWERequationnew}.
 
\par
If it is assumed that there is equal number of stages per treatment and equal allocation across all the active treatments then, as a result, if one is using the same stopping boundary shape one can simply just calculate the FWER. This is because it results in equal pairwise error rate (PWER) for each treatment \citep{BrattonDanielJ2016TIer, Choodari-OskooeiBabak2020Anea, GreenstreetPeter2021Ammp}. This 
 removes the potential issue of time trends with changing allocation ratios. Therefore to find the boundaries one can use a single scalar parameter $a$ with the functions $L_{k}=f(a)$ and $U_{k}=g(a)$ where $f$ and $g$ are the functions for the shape of the upper and lower boundaries respectively. This is similar to the method presented in \cite{MagirrD.2012AgDt}.

%
\subsection{Power}
\label{subsec:power}
When designing a multi-arm trial in which all treatments get tested until they are stopped for futility or superiority, regardless of the other treatments,  different definitions of power could be considered. The power of a study is focused on the probability that the trial results in some or all of the treatments going forward. The sample size of the study is then found to ensure that the chosen power is greater than or equal to some chosen value $1-\beta$.
\par
One may be interested in ensuring that at least one treatment is taken forward from the study. This can be split into two types of power discussed in the literature. The first is the disjunctive power \citep{UrachS.2016Mgsd, Choodari-OskooeiBabak2020Anea, HamasakiToshimitsu2021Ostc} which is the probability of taking at least one treatment forward. The second is the pairwise power which is the probability of taking forward a given treatment which has a clinically relevant effect \citep{Choodari-OskooeiBabak2020Anea, RoystonPatrick2011Dfct}. 
In the Supporting Information Section 2 the equations needed to calculate the disjunctive power ($P_D$) are given. 
\par
Another way of thinking of powering a study is the probability of taking forward all the treatments which have a clinically relevant effect. This is known as the conjunctive power of a study \citep{UrachS.2016Mgsd, Choodari-OskooeiBabak2020Anea, HamasakiToshimitsu2021Ostc, SerraAlessandra2022Aorm}. 
For the conjunctive power we prove that it is lowest when all the treatments have the clinically relevant effect.

\subsubsection{Pairwise power}
The pairwise power of a treatment is independent of other active treatments. This is because the other active treatments effect has no influence on the treatment of interest as these are independent. Therefore we only need to consider the probability that the treatment of interest with a clinically relevant effect is found superior to the control.  The pairwise power for treatment $k$ ($P_{pw,k}$) with the clinically relevant effect $\theta'$ is: 
\begin{equation}
P_{pw,k}=\sum^{J_k}_{j=1} \Phi(U^{+}_{k,j}(\theta'),L^{+}_{k,j}(\theta'),\ddot{\Sigma}_{k,j}), 
\label{PWER}
\end{equation}
with
\begin{align}
\begin{split}
L^{+}_{k,j}(\theta_k)=(l_{k,1}&-\frac{\theta_k}{\sqrt{I_{k,1}} },\hdots, l_{k,j-1}-\frac{\theta_k}{\sqrt{I_{k,j-1} }}, u_{k,j}-\frac{\theta_k}{\sqrt{I_{k,j-1} }} ) 
\end{split}
\label{equ:Lplus}
\\
U^{+}_{k,j}(\theta_k)=(u_{k,1}& -\frac{\theta_k}{\sqrt{I_{k,j-1} }},\hdots, u_{k,j-1} -\frac{\theta_k}{\sqrt{I_{k,j-1} }}, \infty).
\label{equ:Uplus}
\end{align} 
and $
I_{k,j}=\sigma^2( n_{k,j}^{-1}+(n_{0,k,j}-n(k))^{-1})
$. The $(i,i^\star)^\text{th}$ element $(i \leq i^\star)$ of the covariance matrix $\ddot{\Sigma}_{k,j}$  is
\begin{align*}
\Bigg{(}\sqrt{r_{k,i}^{-1}+(r_{0,k,i}-r(k))^{-1}}\sqrt{r_{k,i^\star}^{-1}+(r_{0,k,i^\star}-r(k))^{-1}} \Bigg{)}^{-1}  \Bigg{(} \frac{1}{r_{k,i^\star}} + \frac{1}{r_{0,k,i^\star}-r(k)}  \Bigg{)}.
\end{align*}
One can then design the trial so that the pairwise power for each treatment $k$ ($P_{pw,k}$) is greater than or equal to some chosen $1-\beta$ for every treatment. If one has an equal number of stages per treatment and equal allocation across all the active treatments with the same stopping boundaries, this ensures that pairwise power is equal for each treatment so $n_{k}=n_{k^\star}$ for all $k,k^\star \in 1,\hdots, K$. Therefore we define $n=n_k$ for all $k \in 1,\hdots, K$. To ensure pairwise power is controlled, keep increasing $n$ until $P_{pw} \geq 1-\beta$ where $P_{pw}=P_{pw,k}$ for all $k \in 1,\hdots, K$. 
\par
If one designing a trial in which there is a set number of patients to the control before an active treatment $k$ is added, so $n(k)$ is predefined before calculating the boundaries and sample size, one needs to use an approach such as the Algorithm below. This is because when the sample size increases there is no increase in $n(k)$ for all $k$. This results in a change in the allocation ratio between $r(k)$ and $r_{0,k,j}$ for each $j$. Therefore requiring the bounds to be recalculated for the given $r(k)$. If one focus is on the new arms being added after a set percentage of the way through the trial this issue no longer persists, as the allocation ratio stays the same so the bounds can be calculated ones.

\begin{algorithm}[]
\caption{Iterative approach to compute the $n$ for the pairwise power with predefined $\mathbf{n(K)}$}
\begin{itemize}
\item[0] Begin by assuming $\mathbf{n(K)}=(0,0,\hdots,0)$ and find the stopping boundaries to control the FWER. Now calculate $n$ such that the pairwise power is greater then or equal to a pre-specified ($1-\beta$). Then repeat the following iterative steps until the pairwise power, given the true $\mathbf{n(K)}$, is greater than ($1-\beta$):
\par
\item[1] Find the stopping boundaries to control the FWER for the true predefined $\mathbf{n(K)}$ given the current $n$.
\item[2] Calculate $P_{pw}$ for the given boundaries.
\par
\item[3] If $P_{pw} \geq 1-\beta$ then stop, else increase $n$ by 1 and repeat steps 1-3.
\end{itemize}
\label{Alg:pairwisepower}
\end{algorithm}

\subsubsection{Conjunctive power}
The conjunctive power is defined as the probability of taking forward all the treatments which have a clinically relevant effect.
We begin by proving when the conjunctive power is at its lowest.
We define the events
\begin{align*}
B_{k,j}(\theta_k)=&[ l_{k,j}  + (\mu_k-\mu_0 -\theta_k)I^{1/2}_{k,j} < Z_{k,j}  <  u_{k,j} + (\mu_k-\mu_0 -\theta_k)I^{1/2}_{k,j}],
\\
C_{k,j}(\theta_k)=&[Z_{k,j} > u_{k,j} + (\mu_k-\mu_0 -\theta_k)I^{1/2}_{k,j}],
\end{align*}
where $B_{k,j}(\theta_k)$ defines the event that treatment $k$ continues to the next stage and  $C_{k,j}(\theta_k)$ defines the event that treatment $k$ is found superior to the control at stage $j$.
If $\mu_k -\mu_0=\theta_k$ for $k=1,\hdots, K$, the event that $H_{01},\hdots, H_{0K}$ are all rejected $(\bar{W}_K(\Theta))$ is equivalent to
\begin{align*}
\bar{W}_K(\Theta)
& = \bigcap_{k \in \{m_1, \hdots, m_K \} } \Bigg{(} \bigcup^{J_k}_{j=1} \Bigg{[} \bigg{(} \bigcap^{j-1}_{i=1}B_{k,i}(\theta_k) \bigg{)} \cap C_{k,j}(\theta_k) \Bigg{]} \Bigg{)},
\end{align*}
where $\Theta= \{\theta_1,\theta_2,\hdots,\theta_K\}$.
\begin{theorem}
For any $\Theta$, $P( \text{reject all } H_{0k} \text{ for which }
\theta_k \geq \theta'| \Theta) \geq P(\text{reject all } H_{0k} \text{ for which }
\theta_k \geq \theta'| (\mu_1=\mu_2=\hdots=\mu_K=\mu_0+\theta'))$.
\label{the:TotalPower}
\end{theorem}

\par

\begin{proof}
For any $\epsilon_k<0$,
\begin{align*}
\bigcup^{J_k}_{j=1}  \Bigg{[} \bigg{(} \bigcap^{j-1}_{i=1}B_{k,i}(\theta_k+\epsilon_k) \bigg{)} \cap C_{k,j}(\theta_k +\epsilon_k) \Bigg{]} \subseteq 
\bigcup^{J_k}_{j=1} \Bigg{[} \bigg{(} \bigcap^{j-1}_{i=1}B_{k,i}(\theta_k) \bigg{)} \cap C_{k,j}(\theta_k) \Bigg{]}.
\end{align*}
Take any
$$w=(Z_{k,1},\hdots,Z_{k,J}) \in \bigcup^{J_k}_{j=1}  \Bigg{[} \bigg{(} \bigcap^{j-1}_{i=1}B_{k,i}(\theta_k+\epsilon_k) \bigg{)} \cap C_{k,j}(\theta_k +\epsilon_k) \Bigg{]}.$$
For some $q \in \{1,\hdots, J_k \},$ for which $Z_{k,q} \in C_{k,q} (\theta_k +\epsilon_k)$ and $Z_{k,j} \in B_{k,j} (\theta_k +\epsilon_k)$ for $j=1,\hdots,q-1$. $Z_{k,q} \in C_{k,q} (\theta_k +\epsilon_k)$ implies that $Z_{k,q} \in C_{k,q} (\theta_k)$.
Furthermore $Z_{k,q} \in B_{k,q} (\theta_k +\epsilon_k)$ implies that $Z_{k,q} \in B_{k,q} (\theta_k) \cup C_{k,q} (\theta_k)$ for some $j=1,\hdots, q-1$.
Therefore,
\begin{equation*}
w \in \bigcup^{J_k}_{j=1} \Bigg{[} \bigg{(} \bigcap^{j-1}_{i=1}B_{k,i}(\theta_k) \bigg{)} \cap C_{k,j}(\theta_k) \Bigg{]}.
\end{equation*}

Next suppose for any $m_1, \hdots, m_K$ where $m_1 \in \{1,\hdots, K\} $ and $m_k \in \{1,\hdots, K\} \backslash \{ m_1,\hdots, m_{k-1} \}$ with $\theta_{m_1}, \hdots, \theta_{m_l} \geq \theta'$ and $\theta_{m_{l+1}}, \hdots, \theta_{m_K} < \theta'$. Let $\Theta_l=(\theta_{m_1},\hdots, \theta_{m_l})$. Then
\begin{align*}
P( \text{reject all } H_{0k} \text{ for which }
\theta_k \geq \theta'| \Theta)
& = P(\bar{W}_l (\Theta_l))
\\
& \geq P(\bar{W}_l (\Theta')) 
\\
& \geq P(\bar{W}_k (\Theta')) 
\\
&= P( \text{reject all } H_{0k} \text{ for which }
\theta_k \geq \theta'| H_{PG}).
\end{align*}
where $\Theta'=(\theta',\hdots, \theta').$
\end{proof}

It follows from Theorem \ref{the:TotalPower} that the conjunctive power ($P_{C}$) is minimised when all treatments have the smallest interesting treatment effect. In order to ensure the conjunctive power is greater than level $1-\beta$ we rearrange the events  $B_{k,j}(\theta_k)$ and $C_{k,i}(\theta_k)$ to find
\begin{equation}
P_{C} = P(\bar{W}_l(\Theta'))=\sum_{\substack{j_k=1 \\ k=1,2,\ldots,K}}^{J_k} \Phi(\mathbf{L}^+_{\mathbf{j_k}}(\Theta'),\mathbf{U}^+_{\mathbf{j_k}}(\Theta'),\Sigma_{\mathbf{j_k}}),
\label{Powerequation}
\end{equation}
where $\mathbf{U}^+_{\mathbf{j_k}}(\Theta')=(U^+_{1,j_1}(\theta'), \hdots, U^+_{K,j_K}(\theta'))$ and $\mathbf{L}^+_{\mathbf{j_k}}(\Theta')=(L^+_{1,j_1}(\theta'), \hdots, L^+_{K,j_K}(\theta'))$ with $U^+_{k,j_k}(\theta_k)$ and $L^+_{k,j_k}(\theta_k)$ defined in Equation (\ref{equ:Uplus}) and Equation (\ref{equ:Lplus}),  respectively. The correlation matrix  $\Sigma_{\mathbf{j_k}}$ is the same as that given for FWER in Equation (\ref{equ:FullCor}).  
\par
When one has equal number of stages and equal allocation to find the sample size one needs to increase $n$ until $P_{C}=1-\beta$. If one is in the case of fixed $\mathbf{n(k)}$ then one can use Algorithm \ref{Alg:pairwisepower}, now replacing pairwise power for conjunctive power. 

\subsection{Sample size distribution and Expected sample size}
\label{sec:expsam}
The determination of sample size distribution and expected sample size involves calculating the probability for each outcome of the trial, denoted as $Q_{\mathbf{j_k},\mathbf{q_k}}$. Here, $\mathbf{q_k}=( q_1,\hdots, q_K )$ is defined, where $q_k= 0$ indicates that treatment $k$ falls below the lower stopping boundary at point $j_k$, and $q_k= 1$ indicates that treatment $k$ exceeds the upper stopping boundary at point $j_k$. We find
\begin{align*}
Q_{\mathbf{j_k},\mathbf{q_k}}= & \Phi(\tilde{\mathbf{L}}_{\mathbf{j_k,q_k}}(\Theta),\tilde{\mathbf{U}}_{\mathbf{j_k,q_k}}(\Theta),\Sigma_{\mathbf{j_k}}),
\end{align*}
with $\mathbf{j_k}$ one can define  $\tilde{\mathbf{L}}_{\mathbf{j_k,q_k}}(\Theta)=(\tilde{L}_{1,j_1,q_1}(\theta_1), \hdots, \tilde{L}_{K,j_K,q_K}(\theta_K))$ and $\tilde{\mathbf{U}}(\Theta)_{\mathbf{j_k,q_k}}=(\tilde{U}_{1,j_1,q_1}(\theta_1),$ $ \hdots, \tilde{U}_{K,j_K,q_K}(\theta_K))$ where 
\begin{align*}
\tilde{L}_{k,j,q_k}(\theta_k)=&(l_{k,1}-\frac{\theta_k}{\sqrt{I_{k,1} }},\hdots, l_{k,j-1}-\frac{\theta_k}{\sqrt{I_{k,j-1} }}, [\mathbbm{1}(q_k=0)(-\infty)+ u_{k,j_k}]-\frac{\theta_k}{\sqrt{I_{k,j} }} ),
\\
\tilde{U}_{k,j,q_k}(\theta_k)=&(u_{k,1} -\frac{\theta_k}{\sqrt{I_{k,1} }},\hdots, u_{k,j-1} -\frac{\theta_k}{\sqrt{I_{k,j-1} }}, [\mathbbm{1}(q_k=1)(\infty)+ l_{k,j_k}]-\frac{\theta_k}{\sqrt{I_{k,j} }} ),
\end{align*} 
respectively.
The $P_{\mathbf{j_k},\mathbf{q_k}}$  are associated with their given total sample size $N_{\mathbf{j_k},\mathbf{q_k}}$ for that given $\mathbf{j_k}$ and $\mathbf{q_k}$.
\begin{align*}
N_{\mathbf{j_k},\mathbf{q_k}} = &\bigg{(} \sum^K_{k=1} n_{k,j_k} \bigg{)}
+\max_{k\in 1,\hdots K}(n_{0,k,j_k}),
\end{align*}

This shows that the control treatment continues being recruited to until, at the earliest, the last active treatment to be added has had at least one analysis. To obtain the sample size distribution, as similarly done in \cite{GreenstreetPeter2021Ammp}, we group all the values of $\mathbf{j_k}$ and $\mathbf{q_k}$ that gives the same value of $N_{\mathbf{j_k},\mathbf{q_k}}$ with its corresponding $Q_{\mathbf{j_k},\mathbf{q_k}}$. This set of $Q_{\mathbf{j_k},\mathbf{q_k}}$ is then summed together to give the probability of the realisation of this sample size. To calculate the sample size distribution for each active arm, group $n_{k,j_k}$ with its corresponding $Q_{\mathbf{j_k},\mathbf{q_k}}$ and this can similarly be done for the control treatment. The expected sample size for a given $\Theta$, denoted as $E(N|\Theta)$, is obtained by summing all possible combinations of $\mathbf{j_k}$ and $\mathbf{q_k}$,

\begin{equation}
E(N|\Theta)= \sum_{\substack{j_k=1 \\ k=1,2,\ldots,K}}^{J_k} \sum_{\substack{q_k\in \{ 1,\infty \} \\ k=1,2,\ldots,K}} Q_{\mathbf{j_k},\mathbf{q_k}} N_{\mathbf{j_k},\mathbf{q_k}}.
\label{eq:ESS}
\end{equation}
The expected sample size for multiple different treatment effects ($\Theta=\{\theta_1,\hdots,\theta_K\}$) can then be found using Equation (\ref{eq:ESS}).
\section{Motivating trial example}
\label{sec:Trialexample}
\par
\subsection{Setting}
One example of a platform trial is FLAIR, which focused on chronic lymphocyte leukemia \cite{HowardDenaR2021Apti}. FLAIR initially planned to incorporate an additional active treatment arm and conduct an interim analysis midway through the intended sample size for each treatment. During the actual trial, two extra arms were introduced, including an additional control arm. The original trial design primarily addressed the pairwise type I error due to the inclusion of both additional experimental and control arms.
\par
Following \cite{GreenstreetPeter2021Ammp}, a hypothetical trial that mirrors some aspects of FLAIR will be studied. In this hypothetical trial the familywise error rate (FWER) in the strong sense will be controlled. Controlling the FWER may be seen as crucial in this scenario, as the trial aims to assess various combinations of treatments involving a common compound for all active treatments \citep{WasonJamesMS2014Cfmi}. There is an initial active treatment arm, a control arm, and a planned addition of one more active treatment arm during the trial. We apply the proposed methodology to ensure FWER control and consider the conjunctive power and pairwise power.
\par
The pairwise power is the main focus of the simulation study rather than the disjunctive power, as a potential drawback of disjunctive power is it is highly dependent on the treatment effect of all the treatments in the study, even the ones without a clinically relevant effect. For example assume one treatment has a clinically relevant effect and the rest have effect equal to the control treatment, then the disjunctive power will keep increasing the more treatments that are added, if one keeps the same bounds, even though the probability of taking the correct treatment forward does not increase. Equally the minimum the disjunctive power can be is equal to the pairwise power. This is when only one treatment has a clinically relevant effect and the rest have an extreme negative effect. A further advantage of the pairwise power is it gives the probability of the treatment with the greatest treatment effect being found, assuming that this treatment has effect equal to the clinically relevant effect. 
\par
Considering the planned effect size from FLAIR, we assume an interesting treatment difference of $\theta'=-\log(0.69)$ and a standard deviation of $\sigma=1$. It should be noted that while FLAIR used a time-to-event endpoint with $0.69$ representing the clinically relevant hazard ratio between the experimental and control groups, our hypothetical trial will focus on continuous endpoints using a normal approximation of time-to-event endpoints as discussed in \cite{JakiT2013Coca}. The desired power is $80\%$. We will maintain the same power level as FLAIR while targeting a one-sided FWER of 2.5\%. The active treatment arms interim analysis will be conducted midway through its recruitment and 1:1 allocation will be used between the control and the active treatments as done in FLAIR \citep{HillmenPeter2023Iarv}. 
\par
The difference between a design which controls the pairwise power and the conjunctive power will be studied in Section \ref{subsec:comparingpowers}. Additionally, for both pairwise power and the conjunctive power, the number of patients per arm per stage, the maximum sample size, the expected sample size and the disjunctive power will be studied. In Section \ref{subsec:separatetrials} the effect of different numbers of patients recruited to the control before the second treatment is added ($n(2)$) will be studied with the focus being on expected sample size and maximum sample size of the trial. The designs will be compared to running two completely separate independent trials for each of the 2 active treatments. 
When running two trials there would be less expectation to control the FWER across the two trials. Therefore along with the fair comparison of type I error control of 2.5\% across the multiple separate studies, the setting of having pairwise error rate being controlled for each at 2.5\% will be shown. 
 In Section \ref{subsec:separatetrialstypeIerror} the effect of using a more liberal FWER control compared to type I error control for the separate trials is studied for trials with 3 and 4 active arms.

\subsection{Comparing the two types of power}
\label{subsec:comparingpowers}
We will consider the effect of adding the second treatment halfway through recruitment of the first active treatment, both for ensuring pairwise power and conjunctive power are at 80\%. The binding triangular stopping boundaries will be used \citep{WhiteheadJ.1997TDaA, WasonJamesM.S.2012Odom, li2020optimality}. The stopping boundaries are the same regardless of if one is controlling pairwise power or conjunctive power as $r(2)=r_{1,1}$ for both. The stopping boundaries are given in Table \ref{tab:BoundsSS} and are equal for both designs.
\par
In Table \ref{tab:BoundsSS} the sample size when ensuring that the pairwise power is greater than 80\% is given. Both active treatments will have up to 152 patients recruited to them and the control treatment can have up to 228 patients. This is due to 76 patients already being recruited to the control before the second treatment is added. The maximum sample size for the pairwise power design is therefore $\max(N)=152+152+228=532$. Additionally in Table \ref{tab:BoundsSS} the sample size when ensuring that the conjunctive power is greater than 80\% is given. The maximum sample size now is $\max(N)=192+192+288=672$. The calculations were carried out using R \citep{Rref} with the method given here having the multivariate normal probabilities being calculated using the packages \texttt{mvtnorm} \citep{mvtnorm} and \texttt{gtools} \citep{warnes2021package}. Code is available at https://github.com/pgreenstreet/Platform\_trial\_multiple\_superior.

\begin{table}[H]
\centering
 \caption{ The stopping boundaries and sample size of the proposed designs, for both control of pairwise power and of conjunctive power.}
\resizebox{\textwidth}{!}{\begin{tabular}{c|c|c|c|c|c|c}
Design  & \multirow{2}{*}{ $U = \begin{pmatrix}
U_1 \\
U_2  \\ 
\end{pmatrix}$} & \multirow{2}{*}{$L =
\begin{pmatrix}
L_1 \\
L_2 \\ 
\end{pmatrix}$} & \multirow{2}{*}{$\begin{pmatrix}
n_{1,1}  & n_{1,2} \\
 n_{2,1}  & n_{2,2}   
\end{pmatrix}$ }&  \multirow{2}{*}{$\begin{pmatrix}
n_{0,1,1}  & n_{0,1,2} \\
 n_{0,2,1}  & n_{0,2,2}  
\end{pmatrix}$}   & \multirow{2}{*}{$\begin{pmatrix}
n(1) \\
n(2)
\end{pmatrix}$} & \multirow{2}{*}{$\max(N)$} \\
controlling & & & & &  \\
\hline
Pairwise & \multirow{2}{*}{$\begin{pmatrix}
2.501 & 2.358 \\
2.501 & 2.358  \\ 
\end{pmatrix}$} & \multirow{2}{*}{$\begin{pmatrix}
0.834 & 2.358 \\
0.834 & 2.358 \\ 
\end{pmatrix}$}  & \multirow{2}{*}{$\begin{pmatrix}
 76  & 152  \\
 76  & 152   
\end{pmatrix}$} & \multirow{2}{*}{$\begin{pmatrix}
 76  & 152 \\
 152  & 228
\end{pmatrix}$} & \multirow{2}{*}{$\begin{pmatrix}
 0 \\
 76 
\end{pmatrix}$}  & \multirow{2}{*}{532} \\
power & & & & & &
\\
Conjunctive & \multirow{2}{*}{$\begin{pmatrix}
2.501 & 2.358 \\
2.501 & 2.358  \\ 
\end{pmatrix}$} & \multirow{2}{*}{$\begin{pmatrix}
0.834 & 2.358 \\
0.834 & 2.358 \\ 
\end{pmatrix}$}  & \multirow{2}{*}{$\begin{pmatrix}
 96  & 192  \\
 96  & 192   
\end{pmatrix}$} & \multirow{2}{*}{$\begin{pmatrix}
 96  & 192 \\
 192  & 288
\end{pmatrix}$} & \multirow{2}{*}{$\begin{pmatrix}
 0 \\
 96
\end{pmatrix}$} & \multirow{2}{*}{672} \\  
power & & & & & &

\end{tabular}}
\label{tab:BoundsSS} 
\end{table}

\par
Based on the two designs in Table \ref{tab:BoundsSS}, in Table \ref{tab:resultstri} the conjunctive power, pairwise power and disjunctive power for different values of $\theta_1$ and $\theta_2$ are given along with the expected sample size. The values of $\theta_1$ and $\theta_2$ are chosen to study the effects under the global null, when treatments have a clinically relevant effect and when one of the active treatments performs considerably worst than the rest.
%
Table \ref{tab:resultstri} shows when $\theta_1$ and $\theta_2$  equals the clinically relevant effect $\theta'$ under the design for pairwise power, that the pairwise power of both treatments is 80.0\%; the conjunctive power is 66.0\%; the disjunctive power is 94.1\%; and the expected sample size is 420.6. This highlights the fact that when controlling the pairwise power that if both treatments have a clinically relevant effect there is a large chance (44\%) that one may miss at least one of the two treatments. 
\par
When studying the design in which conjunctive power is controlled one can now see that the pairwise power and disjunctive power is much greater compared to the pairwise power design. This comes with a large increase in both expected and maximum sample size, for example the maximum sample size has increased by 140 patients. 
\par
As seen for the design for conjunctive power section of Table \ref{tab:resultstri} the disjunctive power when treatment 1 and 2 have effect $\theta', 0$, respectively, does not equal the disjunctive power of treatment 1 and 2 when the effect is $0, \theta'$. This is because the outcome of treatment 1's  test statistic has a larger effect on treatment 2 then the other-way round. For example treatment 1 first stage is always independent of treatment 2. However for treatment 2 it's first stage is only independent of treatment 1 if treatment 1 stops after its first stage. Therefore $\Sigma_{(1,2)} \neq \Sigma_{(2,1)}$. However as can be seen this difference in the cases studied is very small.
\par 
In Table \ref{tab:resultstri} shows when there is only one treatment with a clinically relevant effect the conjunctive power equals the pairwise power of that treatment. When neither treatment has a clinically relevant effect the conjunctive power equals 100\%, as there are no treatments with a clinically relevant effect that need to be found.  As a result the trial has already resulted in all the clinically relevant treatments being declare i.e 0 treatments.
\par
The expected sample size is greatly dependent on which treatment has the clinically relevant effect and which does not. For example when studying the the design with pairwise power control the expected sample size when the treatment effect is $\theta', 0$, is 372.7. This is compared to 396.6 when the treatment effect is $0, \theta'$ for treatment 1 and 2 respectively. This difference is because the probability of treatment $k$ stopping after the first stage is higher when $\theta_k=0$ compared to $\theta_k=\theta'$. Therefore when the second treatment has effect 0 it is more likely that the trial will stop after the second stage of the trial. This reduces the amount of patients on average being recruited to the control treatment.
\par
In Table \ref{tab:resultstri} it can be seen that the pairwise power for the treatment with a clinically relevant effect is equal to the disjunctive power when the other treatment has an extremely negative treatment effect compared to the control.  This is as there is no longer a chance that the other treatment can be taken forward. Therefore  $\theta_1=-\infty$ $\theta_2=\theta'$ or $\theta_1=\theta'$ $\theta_2=-\infty$, is the point when the pairwise, disjunctive and conjunctive power are all equal. 
When one treatment has effect $\theta'$ and the other has effect equal to the control the disjunctive power is greater than the pairwise power, as there is still a chance that the other treatment may be taken forward. In Table \ref{tab:resultstri} it is shown that when both treatments have effect $0$ the disjunctive power is equal to the FWER for the trial. In addition when a treatment has effect $0$ this results in the pairwise power for that treatment equalling the PWER.
\begin{table}[]
\centering
 \caption{Operating characteristics of the proposed designs under different values of $\theta_1$ and $\theta_2$, for both control of pairwise power and of conjunctive power.}

\resizebox{\textwidth}{!}{\begin{tabular}{c|c|c|c|c|c|c}

\multicolumn{7}{c}{\textbf{Design for pairwise power}}
\\
\hline
\multicolumn{2}{c|}{Treatment effect} & \multicolumn{2}{c|}{Pairwise power} & Conjunctive power & Disjunctive power & Expected sample size
\\
$\theta_1$ & $\theta_2$ & $P_{PW,1}$ &  $P_{PW,2}$ &  $P_{C}$   & $P_{D}$  & $E(N|\theta_1,\theta_2)$  \\
\hline
 $\theta'$ & $\theta'$ & 0.800  & 0.800 & 0.660 & 0.941 & 420.6   \\
 $\theta'$ & $0$ & 0.800  & 0.013 & 0.800 & 0.802  & 372.7   \\
 $\theta'$ & $-\infty$ & 0.800  & 0 & 0.800 & 0.800  & 342.9   \\
  0 & $\theta'$ & 0.013  & 0.800 & 0.800 & 0.802 &  396.6   \\
  0 & $0$ & 0.013  & 0.013 & 1 & 0.025 &  348.7   \\
 $-\infty$ & $\theta'$ & 0 & 0.800 & 0.800 & 0.800 &  381.7   \\
\hline

\multicolumn{7}{c}{\textbf{Design for conjunctive power}}
\\
\hline
\multicolumn{2}{c|}{Treatment effect} & \multicolumn{2}{c|}{Pairwise power} & Conjunctive power & Disjunctive power & Expected sample size
\\
$\theta_1$ & $\theta_2$ & $P_{PW,1}$ & $P_{PW,2}$ &  $P_{C}$   & $P_{D}$ & $E(N|\theta_1,\theta_2)$  \\
\hline
 $\theta'$ & $\theta'$ & 0.890  & 0.890 & 0.801 & 0.979 &  508.1   \\
 $\theta'$ & $0$ & 0.890  & 0.013 & 0.890 & 0.890  & 463.0   \\
 $\theta'$ & $-\infty$ & 0.890  & 0 & 0.890 & 0.890  & 425.4   \\
  0 & $\theta'$ & 0.013  & 0.890 & 0.890 & 0.891  & 485.6   \\
  0 & $0$ & 0.013  & 0.013 & 1 & 0.025 &  440.5   \\
 $-\infty$ & $\theta'$ & 0 & 0.890 & 0.890 & 0.890  & 466.7  \\
\end{tabular}}
\label{tab:resultstri} 
\end{table}

In the Supporting Information Section 3 results for using both O'Brien and Fleming \citep{OBrienPeterC.1979AMTP} and Pocock boundaries \citep{PocockStuartJ.1977GSMi} are shown, with the futility boundary equal to 0 \citep{MagirrD.2012AgDt}. Additionally the results for using non binding triangular stopping boundaries are shown in the Supporting Information Section 4.
Overall Table \ref{tab:BoundsSS} and Table \ref{tab:resultstri} have shown that the choice of type of power to control may be highly dependent on the sample size available, as if the design ensures conjunctive power of level $1-\beta$ it will ensures pairwise power of at least $1-\beta$ but the opposite does not hold. However the sample size for a trial designed for pairwise power will be less than that of a design for conjunctive power.    

\subsection{Comparison with running separate trials}
\label{subsec:separatetrials}
This section studies the effect on maximum and expected sample size depending on when the additional treatment arm is added to the platform trial. The examples for both conjunctive power and pairwise power are compared to running two separate trials. There are two settings for separate trials which are considered. Setting 1 is when the type I error across both the trials is set to be 2.5\%, therefore, the type I error for each is $1-\sqrt{1-0.025}=1.26\%$. For Setting 2 the type I error of each trial is controlled at 2.5\%. For the separate trials which are compared to the pairwise power, the power level for each is set to 80\%. This results in the following sample size and stopping boundaries for the two trials for Setting 1, 
\begin{equation*}
U_1=
\begin{pmatrix}
2.508 & 2.364 \\ 
\end{pmatrix}, 
\:\:\:\:
L_1 =
\begin{pmatrix}
0.836 & 2.364 \\
\end{pmatrix}.
\:\:\:\:
\begin{pmatrix}
n_{1,1} & n_{1,2} \\
\end{pmatrix}=\begin{pmatrix}
77 & 154 \\
\end{pmatrix}.
\end{equation*}  
with $n_{0,1,1}=n_{1,1}$, $n_{0,1,2}=n_{1,2}$ and $n(1)=0$. Setting 2 gives:
\begin{equation*}
U_1=
\begin{pmatrix}
2.222 & 2.095 \\ 
\end{pmatrix}, 
\:\:\:\:
L_1 =
\begin{pmatrix}
0.741 & 2.095 \\
\end{pmatrix}.
\:\:\:\:
\begin{pmatrix}
n_{1,1} & n_{1,2} \\
\end{pmatrix}=\begin{pmatrix}
65 & 130 \\
\end{pmatrix}.
\end{equation*}  
with $n_{0,1,1}=n_{1,1}$, $n_{0,1,2}=n_{1,2}$ and $n(1)=0$. For comparison with the conjunctive power designs the probability of finding both treatments across the two trials is set to 80\%. The required power for each trial is therefore $\sqrt{1-\beta}=0.894$. The boundaries remain the same for both settings  as the type I error remains the same. The new sample size for Setting 1 is $\begin{pmatrix}
n_{1,1} & n_{1,2} \\
\end{pmatrix}=\begin{pmatrix}
98 & 196 \\
\end{pmatrix}$  and for Setting 2 is $\begin{pmatrix}
n_{1,1} & n_{1,2} \\
\end{pmatrix}=\begin{pmatrix}
85 & 170 \\
\end{pmatrix}$.

Figure \ref{fig:PWplot} gives the maximum sample size and the expected sample size under different $\theta_1, \theta_2$ depending on when the second treatment is added, for the pairwise power control of 80\%. 
 Figure \ref{fig:TPplot} gives similar results however the focus now is on control of the conjunctive power at 80\%.

As indicated in Figure \ref{fig:PWplot}, when controlling the pairwise power, if the second active treatment is introduced at the beginning of the trial, the total sample size required is $456$, whereas if it is added at the end of recruitment for treatment 1, the total sample size becomes $616$. This increase in sample size is attributable to two factors. Firstly, there is a necessity to increase the number of patients recruited to the control group until treatment 2 has completed the trial. Secondly, the decrease in correlation between the two treatments results in an enlargement of the boundaries to maintain control over the family-wise error rate. It is this secondary factor which causes the small jumps in maximum sample size seen in Figures \ref{fig:PWplot} and \ref{fig:TPplot}. 
\par
In Figure \ref{fig:PWplot} when comparing the platform designs with pairwise power control, to running two separate trials it can be seen that, for the case that the pairwise error for each trial is 2.5\%, once the second treatment is added after 64 patients have been recruited to the control ($n(2)\geq 64$)  the maximum sample size of running the platform design is greater than or equal to that of running two separate trials, which is $520$ patients. However when controlling the error across both separate trials the maximum sample size is now the same as when adding the second treatment at the end of recruitment for the first treatment in the platform design so $616$. For Setting 1 it can be seen that the expected sample size for separate trials can be better than that of the platform design. In the case of $\theta_1=-\infty$ and $\theta_2=-\theta'$ then once $n(2)\geq 81$  the expected sample size of running the platform design is greater than that of running two separate trials. This is because in the platform approach the control cannot stop until each treatment has finished testing, whereas in the separate trial case each control group will stop as soon as either treatment is dropped. For Setting 1 there are some cases studied which cannot be seen in Figure \ref{fig:PWplot}. These are $\theta_1=\theta'$, $\theta_2=\theta'$ and if $\theta_1=\theta'$, $\theta_2=0$ as both these are at the point $n(2)\geq 117$ which matches that of $\theta_1=\theta'$, $\theta_2=-\infty$. 
When studying the expected sample size of the Setting 2 compared to the platform designs it can be seen that if $\theta_1=-\infty$ and $\theta_2=-\theta'$ then once $n(2)\geq 15$  the expected sample size of running the platform design is greater than that of running two separate trials. The expected sample size for two separate trials when $\theta_1=-\infty$ and $\theta_2=\theta'$ is $319.5$.
\par
When controlling the conjunctive power, as in Figure \ref{fig:TPplot}, if the second active treatment is introduced at the beginning of the trial, the total sample size required is $558$, whereas if it is added at the end of recruitment for treatment 1, the total sample size becomes $784$. Once again the maximum sample size for Setting 1 equals that of when treatment 2 is added after treatment 1 finished recruitment so $784$ patients.
In Figure \ref{fig:TPplot}, when $n(2)\geq 104$  the maximum sample size of running the platform design is greater than or equal to that of running two separate trials under Setting 2, which is $680$ patients. Similar as seen in Figure \ref{fig:PWplot} there is some lines which overlap for Setting 1 in Figure \ref{fig:TPplot} as $n(2)= 143$ is the point for both $\theta_1=\theta'$, $\theta_2=\theta'$ and $\theta_1=\theta'$, $\theta_2=-\infty$, also $n(2)= 121$ is the point for both $\theta_1=0$, $\theta_2=\theta'$ and $\theta_1=0$, $\theta_2=0$. When $n(2)\geq 104$ for Setting 1, and $n(2)\geq 39$ for Setting 2, the expected sample size of running the platform design is greater than that of running two separate trials when $\theta_1=-\infty$ and $\theta_2=\theta'$. The expected sample size for running two separate trials when $\theta_1=-\infty$ and $\theta_2=\theta'$ is $475.3$ and $403.8$ for Setting 1 and Setting 2 respectively.

\begin{figure}[]
  \centering
  \includegraphics[width=.49\linewidth,trim= 0 0.5cm 0 0cm, clip]{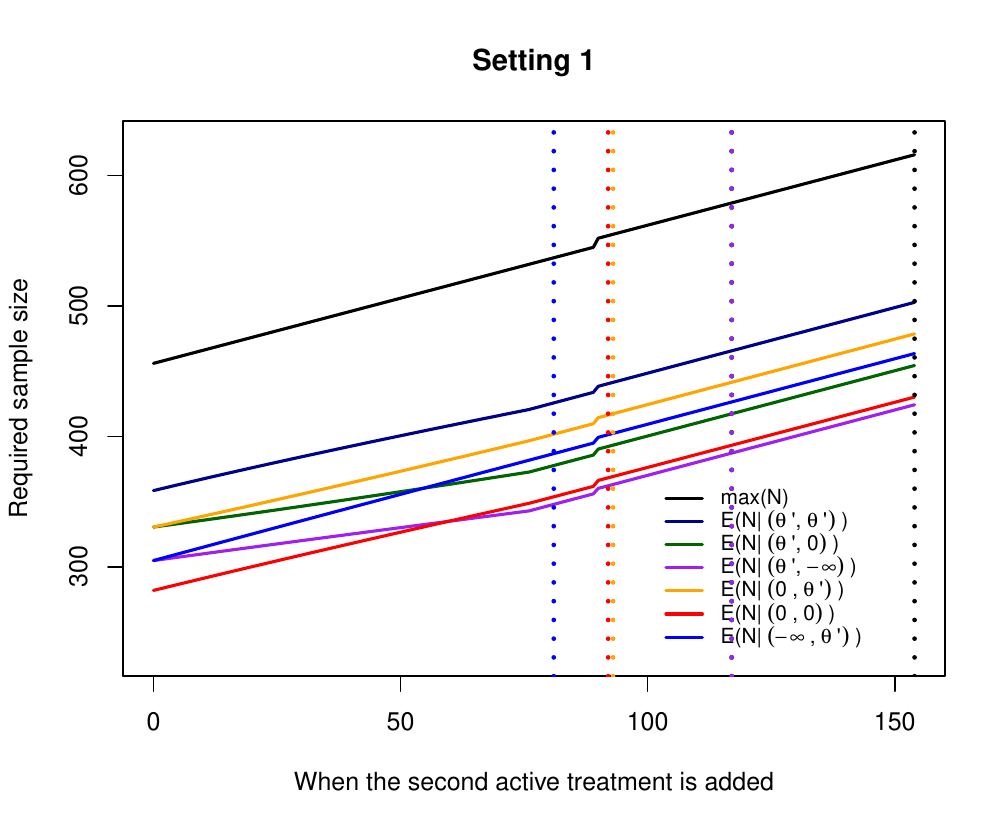}  
  \centering
  \includegraphics[width=.49\linewidth,trim= 0 0.5cm 0 0cm, clip]{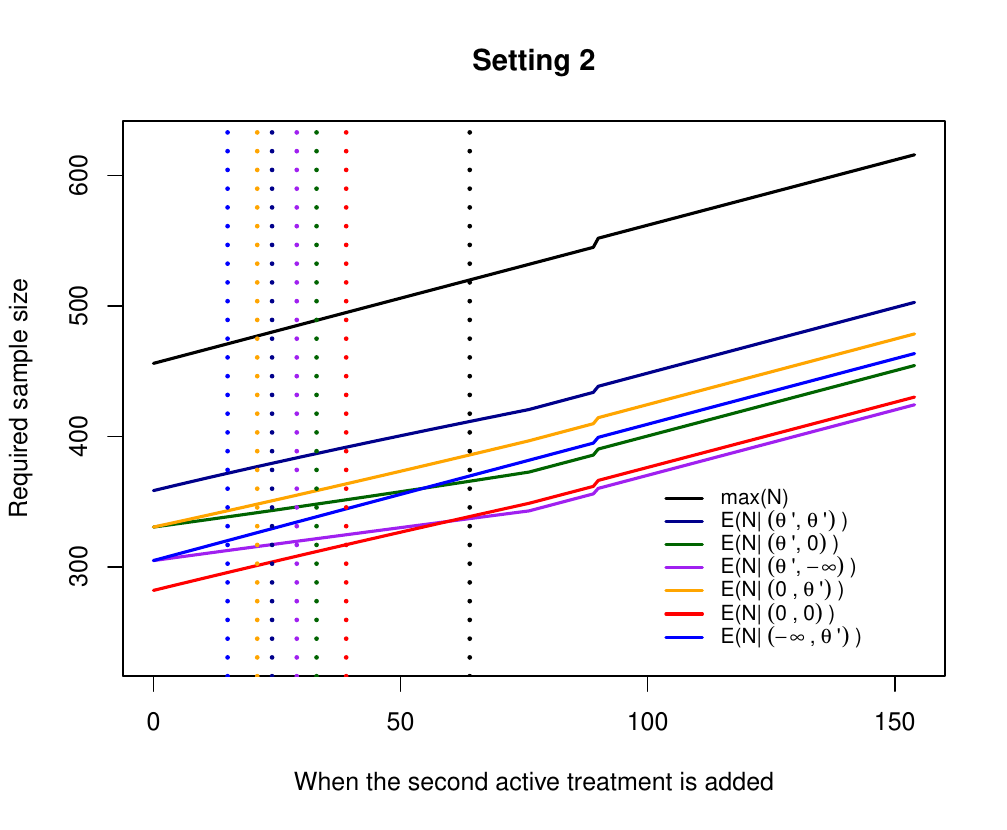}  
\caption{Both panels give the maximum sample size and the expected sample size under different $\theta_1, \theta_2$ depending on the value $n(2)$, for the pairwise power control of 80\%. Left panel: dash vertical lines correspond to the points where the maximum/expected sample size of the trial is now greater than running two separate trials with type I error control across both trials set to 2.5\%. Right panel: dash vertical lines correspond to the points where the maximum/expected sample size of the trial is now greater than running two separate trials with type I error control for each trial set to 2.5\%.}
\label{fig:PWplot}
\end{figure}

\begin{figure}[]
  \centering
  \includegraphics[width=.49\linewidth,trim= 0 0.5cm 0 0cm, clip]{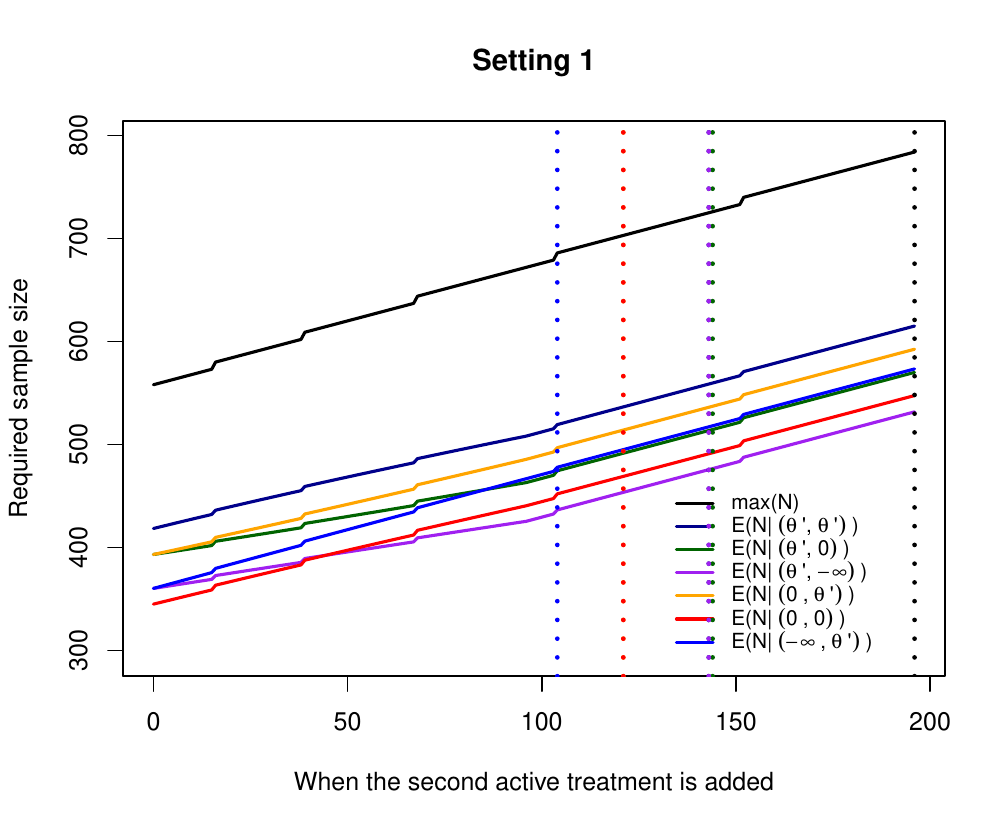}  
  \centering
  \includegraphics[width=.49\linewidth,trim= 0 0.5cm 0 0cm, clip]{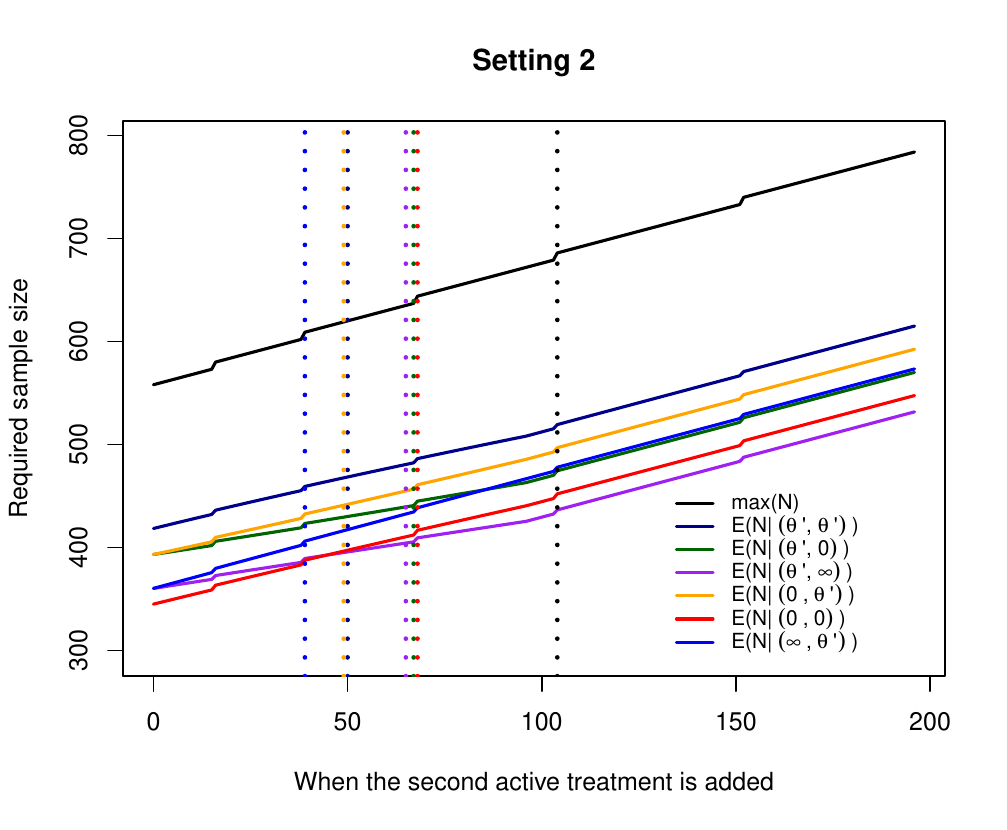}  
\caption{The maximum sample size and the expected sample size under different $\theta_1, \theta_2$ depending on the value $n(2)$, for the conjunctive power control of 80\%. Left panel: dash vertical lines correspond to the points where the maximum/expected sample size of the trial is now greater than running two separate trials under Setting 1. Right panel: dash vertical lines correspond to the points where the maximum/expected sample size of the trial is now greater than running two separate trials under Setting 2.}
\label{fig:TPplot}
\end{figure}
Overall Figures \ref{fig:PWplot} and \ref{fig:TPplot} have shown  
there maybe times that there is no benefit to running a platform trial with regards to sample size, depending on when the later treatment is added to the trial. This issue is further emphasised when there is not the expectation to control the type I error across all the individual trials as seen in Setting 2.
\subsection{Comparison with running separate trials under different controls of type I error}
\label{subsec:separatetrialstypeIerror}
When designing a multi-arm trial one may find that the expected control of the FWER is less than that of the type I error control for an individual trial, as seen in the TAILoR trial for example \citep{PushpakomSudeepP2015TaIR, PushpakomSudeep2020TTaI}. Therefore in Table \ref{tab:FWERvsPWER} we consider the effect of allowing FWER control of 5\% one sided compared to 2.5\% type I error for the individual trials. In this table the same design parameters where used as above, however, now the number of active arms has increased in the hypothetical trial to 3 or 4, and the number of stages is now either 1,2 or 3. In Table \ref{tab:FWERvsPWER} the focus is on controlling the power at the desired 80\% level with the pairwise power being the focus for the top half and conjunctive power for the bottom half. When controlling the conjunctive power the power for each separate trial is $(1-\beta)^{1/k}$.
In these hypothetical trials it is assumed that each one of the arms is added sequentially, with an equal gap between each one. Therefore in the 3 active arm case if the second arm is added after 20 patients have been recruited to the control then the third arm will be added after a total of 40 patients have been recruited to the control.
\par
In Table \ref{tab:FWERvsPWER} the first 2 columns give the number of active arms and stages for the platform trial, respectively. The third and forth columns then gives the sample size per stage and the maximum sample size of the individual trials, respectively. This has been chosen as this number will remain constant throughout, as it is unaffected by the timing of when the next arm is ready due to each trial being completely separate from the others. The remaining columns give when there is no benefit with regards to the maximum and expected sample size of conducting a platform trial compared to running separate trials, with respect $n(k)-n(k-1)$. The value of $n(k)-n(k-1)=n(2)$ as the first treatment is added at the beginning of the trial. In the Supporting Information Section 5 the plots for the 2 stage and 3 stage example trials as given in Table \ref{tab:FWERvsPWER} are shown. 
\par
Using Table \ref{tab:FWERvsPWER}, for the 3 active arm, 2 stage example each separate trial has $n_{1,1}=65$ and $n_{1,2}=130$. The total maximum sample size of running these 3 separate trials is therefore 780. Once the second treatment is planned to be added after 105 patients recruited to the control, (therefore 210 recruited to the control before treatment 3), there is no benefit in using the platform design with respect to maximum sample size. For the expected sample size four different configurations of the treatment effects are studied. The first ($\Theta_1$) assumes all the treatments have the clinically relevant effect, so $\theta_k=\theta'$ for $k=1,\hdots,K$. The second ($\Theta_2$) assumes only the first treatment has a clinically relevant effect and the rest have effect equal to that of the control treatment, so $\theta_1=\theta'$, $\theta_k=0$ for $k=2,\hdots,K$. The third ($\Theta_3$) assumes only the last treatment has a clinically relevant effect and the rest equal the control, so $\theta_K=\theta'$, $\theta_k=0$ for $k=1,\hdots,K-1$. The forth configuration ($\Theta_4$) assumes all the treatments have effect equal to that of the control treatment, so the global null, so $\theta_k=0$ for $k=1,\hdots,K$. For the expected sample size for the 4 treatment effect configurations studied here there is no benefit in using a platform trial after potentially just 62 patients if $\Theta_3$ is true, this does rise to 73 if $\Theta_1$ is true, if the focus is on sample size. 
\par 
Table \ref{tab:FWERvsPWER} shows that the maximum sample size of running separate trials increases with increase in number of stages or arms. This is also the case when running the proposed platform trial design. As can be seen with respect to maximum sample size the more stages the trial has the later a treatment can be added before the maximum sample size becomes worst than running separate trials. For example, when pairwise power is controlled, a 1 stage 3 arm trial with regards to maximum sample size one should use separate trials after 90 patients this is compared to 114 patients for a 3 arm 3 stage trial. 
\par
If the focus is on the expected sample size, than for the examples studied here, increasing the number of stages results in a decrease time before one would switch to separate trials. For example when controlling the conjunctive power, for the 4 arm trial, it can be seen that the expected sample size under the global null for running separate trials becomes less than that of running the platform trials when $n(2)=140$ for 1 stage case compared to $n(2)=99$ for the 3 stage version. This is because the ability to have interim analyses saves more patients for separate trials with respect to expected sample size. This is because in separate trials when a treatment is stopped earlier either for futility or superiority the control treatment also stops. Therefore in this 4 arm example there are 4 sets of control treatments which can stop early compared to only 1 set for the platform design. Additionally for the platform trial the control can only stop once all the active treatments have stopped. This is why the expected sample size under $\Theta_2$ is less then that of $\Theta_3$, as if the final treatment has a clinically relevant effect then it will on average have more stages than a treatment with effect equal to that of the control for the configuration studied here. 
\par
This section has therefore shown that there are periods in which using a platform trial can be beneficial with regards to sample size if one can use a more liberal type I error control compared to that used for individual trials. However this has also shown that if treatments are added late into the trial there may not be benefit, so highlighting the importance of considering which trial design should be use.
\par
\begin{table}[]
\centering
 \caption{The comparison of using the proposed platform design with FWER of 5\% one sided against running separate trials with type I error control of each at 2.5\% one sided, for different numbers of arms and stages.  
 }

\resizebox{\textwidth}{!}{\begin{tabular}{c|c|c|c|c|c|c|c|c}

\multicolumn{8}{c}{\textbf{Design for pairwise power}}
\\
\hline
Active arms & Stages & \multicolumn{2}{c}{Separate trial} \vline & $\min_{n(2)}(\max(N_s)$ & \multicolumn{4}{c}{$\min_{n(2)}(E(N_s | \Theta) \leq E(N | \Theta))$}  \\

$K$ & $J$ & $(n_{1,1} ,\hdots, n_{1,J})$ & $\max(N_s)$  & $\leq \max(N))$ &  $\Theta_1$ &  $\Theta_2$    & $\Theta_3$  & $\Theta_4$   \\
\hline

 3 & 1 & 115  & 690 & 90 & 90 & 90 & 90 & 90   \\
 3 & 2 & (65, 130)  & 780 & 105 & 73 & 72 & 62 & 66  \\
 3 & 3 & (46,  92, 138)  & 828 & 114 & 68 & 67 & 55 & 60   \\
  4 & 1 & 115  & 920 & 79 & 79 & 79 & 79 & 79   \\
  4 & 2 & (65, 130)  & 1040 & 94 & 61 & 62 & 54 & 59    \\
 4 & 3 & (46,  92, 138) & 1104 & 103 & 59 & 58 & 49 & 55   \\
\hline

\multicolumn{8}{c}{\textbf{Design for conjunctive power}}
\\
\hline
Active arms & Stages & \multicolumn{2}{c}{Separate trial} \vline & $\min_{n(2)}(\max(N_s)$ & \multicolumn{4}{c}{$\min_{n(2)}(E(N_s | \Theta) \leq E(N | \Theta))$}  \\

$K$ & $J$ & $(n_{1,1} ,\hdots, n_{1,J})$ & $\max(N_s)$  & $\leq \max(N))$ &  $\Theta_1$ &  $\Theta_2$    & $\Theta_3$  & $\Theta_4$   \\
\hline
 
 3 & 1 & (171)  & 1026 & 143 & 143 & 143 & 143 & 143   \\
 3 & 2 & (97, 194)  & 1164 & 166 & 107 & 109 & 101 & 106   \\
 3 & 3 & (68, 136, 204)  & 1224 & 174 & 98 & 99 & 92 & 98   \\
  4 & 1 & (185)  & 1480 & 140 & 140 & 140 & 140 & 140   \\
  4 & 2 & (105, 210)  & 1680 & 167 & 102 & 109 & 103 & 109   \\
 4 & 3 & (74, 148, 222) & 1776 & 182 & 93 & 99 & 93 & 99    \\
\end{tabular}}
Key: $N_s$ is the sample size of running K separate trials, $\Theta_1$: $\theta_k=\theta'$ for $k=1,\hdots,K$; $\Theta_2$: $\theta_1=\theta'$, $\theta_k=0$ for $k=2,\hdots,K$ ; $\Theta_3$: $\theta_K=\theta'$, $\theta_k=0$ for $k=1,\hdots,K-1$; $\Theta_4$: $\theta_k=0$ for $k=1,\hdots,K$.
\label{tab:FWERvsPWER} 
\end{table}

\section{Discussion}
\label{sec:Discussion}
%

This paper has built on the work of \cite{GreenstreetPeter2021Ammp} to show how one can control the FWER for a trial in which the treatments can be preplanned to be added at any point. This work has then studied the different approaches for powering the trial in which the trial will continue even if a superior treatment is found. This paper shows how the expected sample size and sample size distribution can be found. Finally a hypothetical trial, motivated by FLAIR \citep{HowardDenaR2021Apti} is discussed. This section evaluates the pairwise and conjunctive power when the second active treatment is added halfway through recruitment for the first active treatment. We investigate the operating characteristics for multiple values of $\theta_1$ and $\theta_2$. Then the section goes onto study the effect of adding the later treatments at different points in the platform design and compares these trial designs to running separate trials.
\par
The designs flexibility to incorporate the addition of treatments at any point during a trial allows for the creation of multiple designs that depend on when the treatments are introduced. This approach works effectively until the completion of the initial stage for the treatment that initiated the trial. Up to this point, the treatments can be added when it becomes available, and the boundaries can be set accordingly. However, if the treatments are not ready until after the first analysis, two options can be pursued to avoid bias resulting from knowledge of the first stage results. Firstly, one can choose not to plan for the addition of the treatments and conduct separate trials. As demonstrated in Section \ref{sec:Trialexample}, this approach may require fewer patients overall. Alternatively, one can predefine the times at which the treatments will be added and utilize the corresponding bounds. A drawback here is that if the treatments are not ready by the predefined points, they cannot be added. Nevertheless, for the remaining treatments, the control of family-wise error rate will be maintained.  Due to the bounds being designed to
control FWER across all the hypotheses, therefore, by not adding a treatment and so removing a hypothesis this reduces
the  maximum value of the FWER.
\par
This paper has highlighted a potential issue of increased expected and maximum sample size when requiring strong control of family-wise error rate for a platform trial in which an arm is added later. If one would run two completely separate trials the FWER across the trials would likely not be expected. As a result there is a lot of time where there is no benefit to the platform trial design with regards to maximum or expected sample size as was shown Figure \ref{fig:PWplot} and Figure \ref{fig:TPplot} for Setting 2. This point has been further emphasised in Table \ref{tab:FWERvsPWER} which shows that even with a more liberal FWER control compared to the type I error control of each individual trial there are still many points where one may be better of running separate trials with respect to sample size. This work therefore reiterates the importance of the discussions around type I error control in platform trials \citep{MolloySleF.2022Maip, WasonJamesMS2014Cfmi, WasonJames2016Srfm, HowardDenaR2018Romt, ProschanMichaelA2000PGfM, PROSCHANMA1995MCWC, NguyenQuynh2023PTtI}.
\par
If one instead wants to control the pairwise error, as done for example in STAMPEDE \citep{SydesMatthewR2009Iiam}, one can use Equation (\ref{PWER}), now replacing $\theta'$ with 0. An additional advantage of using the PWER, if controlling the pairwise power, is that the stopping boundaries and the sample size required for each active arm are independent of when the arm is added. Therefore the only change will be how many patients need to be recruited to the control. However one may be find PWER in a platform trial insufficient for error control \citep{WasonJamesMS2014Cfmi, MolloySleF.2022Maip} and may not meet the regulators requirements. 



\par
Building upon this research, a study could be conducted to investigate the impact of having different numbers of stages and stopping boundaries while maintaining equal power and type I error for each treatment, utilizing the approach described in Section \ref{sec:Setting}. However, such an investigation would likely require multiple changes in the allocation ratio, resulting in potential issues with time trends. One could therefore examine methods to handle these  time trends, as explored in \cite{lee2020including, MarschnerIanC2022Aoap, RoigMartaBofill2023Oasi, GreenstreetPeter2021Ammp}. Furthermore a change in allocation ratio between treatments can result in different PWER and pairwise power for each treatment if using the same boundaries for each treatment therefore one could use an interative approach such as that discussed in \cite{GreenstreetPeter2021Ammp}. 
Equally one could study the effect of using non-concurrent controls, but once again this can face a large issue with time trends. The main issue with these time trends is that they are unknown. However one could look into incorporating approaches to reduce the bias potentially caused \citep{lee2020including, MarschnerIanC2022Aoap,WangChenguang2022ABmw,SavilleBenjaminR2022TBTM}.
\par
In Section \ref{subsec:separatetrialstypeIerror} it was assumed for the multi-arm trials that each treatment was added after an equal number of control treatments were recruited so $n(k)-n(k-1)=n(2)$ for $k=2,\hdots,K$. This may however not be the case. One may therefore wish to consider the effect of having multiple treatments beginning the study and then adding additional treatments later. 
The methodology presented in Section \ref{sec:Setting} allows for these changes. However when it comes to the comparison designs there are now multiple options that can be chosen. As done in Section \ref{subsec:separatetrialstypeIerror} one could use separate trials for each comparison, however one could consider using multiple MAMS trials where all treatments begin at once, or a mix of the two. Further points to be considered here is how one can evenly share the power across all these trial types, especially if the focus is on conjunctive power, and also how the type I error should be defined for each comparison trial.
\par
Furthermore, this work could be expanded to incorporate adaptive boundaries that adjust once a treatment is deemed effective, as discussed in \cite{UrachS.2016Mgsd} for the case of multi-arm multi-stage (MAMS) trials. However, such an adaptation would result in a less pre-planned design so  potential further complications in understanding for the clinicians, the trial statisticians and the patients. Additionally, determining the point at which the conjunctive power is at its lowest may no longer be feasible, as dropping each arm would lead to lower bounds for the remaining treatments, thus affecting the conjunctive power assessment. This adaptive approach will likely result in uneven distribution of errors across the treatments added at different points. If one was to then adjust for this one may encounter issues with time trends as the allocation ratio may need to change mid trial.
\par
This paper has given a general formulation for designing a preplanned platform trial with a normal continuous endpoint, and using the work of \cite{JakiT2013Coca} one could apply this methodology to other endpoint such as time-to-event used in FLAIR \citep{HowardDenaR2021Apti}. When using this approach one should be aware of computational issues from calculating high dimensional multivariate normal distributions, if one has a large number of arms and stages in the trial design. If this is an issue then one can restrict to only adding arms at the interims so one can utilise the method of \cite{DunnettCharlesW1955AMCP} as discussed in \cite{MagirrD.2012AgDt, GreenstreetPeter2021Ammp}. 

\subsection*{Acknowledgements}
This report is independent research supported by the National Institute for Health Research (NIHR300576). The views expressed in this publication are those of the authors and not necessarily those of the NHS, the National Institute for Health Research or the Department of Health and Social Care (DHSC). TJ and PM also received funding from UK Medical Research Council (MC\_UU\_00002/14 and MC\_UU\_00002/19, respectively). 
This paper is based on work completed while PG was part of the EPSRC funded STOR-i centre for doctoral training (EP/S022252/1).  For the purpose of open access, the author has applied a Creative Commons Attribution (CC BY) licence to
any Author Accepted Manuscript version arising.

\subsection*{Conflict of Interest}
The authors declare no potential conflict of interests. Alun Bedding is a shareholder of Roche Products Ltd.

\bibliography{refp3}

\newpage

\beginsupplement
\begin{center}
\textbf{\Large Supporting Information: A preplanned multi-stage platform trial for discovering multiple superior treatments with control of FWER and power}
\end{center}
\section{Proof of FWER}
\label{SI:proofFWER}

\par
 As in \cite{MagirrD.2012AgDt} we define for any vector of constants $\Theta=(\theta_1, \hdots, \theta_K)$ and $k=1,\hdots,K,$  $j=1,\hdots, J_k$,  then define the events,
\begin{align*}
A_{k,j}(\theta_k)=&[Z_{k,j}<  l_{k,j} + (\mu_k-\mu_0 -\theta_k)I^{1/2}_{k,j}],
\\
B_{k,j}(\theta_k)=&[ l_{k,j}  + (\mu_k-\mu_0 -\theta_k)I^{1/2}_{k,j} < Z_{k,j}  <  u_{k,j} + (\mu_k-\mu_0 -\theta_k)I^{1/2}_{k,j}].
\end{align*}
The FWER under the equal to
\begin{equation}
1-P(\bar{R}_K(\Theta))
 = 1-P(\bigcap_{k \in \{m_1, \hdots, m_K \} } \Bigg{(} \bigcup^{J_k}_{j=1} \Bigg{[} \bigg{(} \bigcap^{j-1}_{i=1}B_{k,i}(\Theta) \bigg{)} \cap A_{k,j}(\Theta) \Bigg{]} \Bigg{)})
\label{eq:1mpr}
\end{equation}
where if $\mu_k-\mu_0 = \theta_k$ for $k = 1, \hdots, K$, the event that $H_{01},\hdots,H_{0K}$ all fail to be rejected is equal to $\bar{R}_K(\Theta)$. The convention that $\bigcap^{0}_{i=1}= \Omega$ where $\Omega$ is the whole sample space is used and $m_1 \in \{1,\hdots, K\} $ and $m_k \in \{1,\hdots, K\} \backslash \{ m_1,\hdots, m_{k-1} \}$. Therefore $\{m_1,\hdots, m_K\}=\{1,\hdots, K\}$. This notation reflects the fact that the order in which treatments are added affects the FWER as seen in \cite{GreenstreetPeter2021Ammp}.

\begin{theorem}
For any $\Theta$, under the conditions above, $P( \text{reject at least one true } H_{0k} | \Theta) \leq$   \newline
 $P(\text{reject at least one true } H_{0k} | H_{G})$.
\label{the:FWER}
\end{theorem}

\begin{proof}
If $\mu_k -\mu_0=\theta_k$ for $k=1,\hdots, K$, the event that $H_{01},\hdots, H_{0K}$ all fail to be rejected is equivalent to
\begin{align*}
\bar{R}_K(\Theta)
& = \bigcap_{k \in \{m_1, \hdots, m_K \} } \Bigg{(} \bigcup^{J_k}_{j=1} \Bigg{[} \bigg{(} \bigcap^{j-1}_{i=1}B_{k,i}(\theta_k) \bigg{)} \cap A_{k,j}(\theta_k) \Bigg{]} \Bigg{)}.
\end{align*}

Then for any $\epsilon_k>0$,
\begin{align*}
\bigcup^{J_k}_{j=1}  \Bigg{[} \bigg{(} \bigcap^{j-1}_{i=1}B_{k,i}(\theta_k+\epsilon_k) \bigg{)} \cap A_{k,j}(\theta_k +\epsilon_k) \Bigg{]} \subseteq 
\bigcup^{J_k}_{j=1} \Bigg{[} \bigg{(} \bigcap^{j-1}_{i=1}B_{k,i}(\theta_k) \bigg{)} \cap A_{k,j}(\theta_k) \Bigg{]}.
\end{align*}
Take any
$$w=(Z_{k,1},\hdots,Z_{k,J}) \in \bigcup^{J_k}_{j=1}  \Bigg{[} \bigg{(} \bigcap^{j-1}_{i=1}B_{k,i}(\theta_k+\epsilon_k) \bigg{)} \cap A_{k,j}(\theta_k +\epsilon_k) \Bigg{]}.$$
For some $q \in \{1,\hdots, J_k \},$ for which $Z_{k,q} \in A_{k,q} (\theta_k +\epsilon_k)$ and $Z_{k,j} \in B_{k,j} (\theta_k +\epsilon_k)$ for $j=1,\hdots,q-1$. $Z_{k,q} \in A_{k,q} (\theta_k +\epsilon_k)$ implies that $Z_{k,q} \in A_{k,q} (\theta_k)$.
Furthermore $Z_{k,q} \in B_{k,q} (\theta_k +\epsilon_k)$ implies that $Z_{k,q} \in B_{k,q} (\theta_k) \cup A_{k,q} (\theta_k)$ for some $j=1,\hdots, q-1$.
Therefore,
\begin{equation*}
w \in \bigcup^{J_k}_{j=1} \Bigg{[} \bigg{(} \bigcap^{j-1}_{i=1}B_{k,i}(\theta_k) \bigg{)} \cap A_{k,j}(\theta_k) \Bigg{]}.
\end{equation*}

Next suppose for any $m_1, \hdots, m_K$ where $m_1 \in \{1,\hdots, K\} $ and $m_k \in \{1,\hdots, K\} \backslash \{ m_1,\hdots, m_{k-1} \}$ with $\theta_{m_1}, \hdots, \theta_{m_l} \leq 0$ and $\theta_{m_{l+1}}, \hdots, \theta_{m_K} > 0$. Let $\Theta_l=(\theta_{m_1},\hdots, \theta_{m_l})$. Then
\begin{align*}
P(& \text{reject at least one true } H_{0k} | \Theta)
\\
& = 1- P(\bar{R}_l (\Theta_l))
\\
& \leq 1- P(\bar{R}_l (0)) 
\\
& \leq 1- P(\bar{R}_K (0)) 
\\
&= P(\text{reject at least one true } H_{0k} | H_{0.G}).
\end{align*}
\end{proof}
The following proof was nearly identical to the one presented in \cite{GreenstreetPeter2021Ammp} and builds on the work of \cite{MagirrD.2012AgDt}. The only change from \cite{GreenstreetPeter2021Ammp} is now is $P( \text{reject at least}$ $\text{ one true } H_{0k} | \Theta) = 1- P(\bar{R}_l (\Theta_l))$ instead of being $P( \text{reject at least one true } $ $ H_{0k} | \Theta) \leq 1- P(\bar{R}_l (\Theta_l))$.
\section{Disjunctive power}
\label{SI:dispower}
As discussed in Section 2.3 of the main paper the disjunctive power is the probability of taking at least one treatment forward. The disjunctive power can therefore be calculated in a very similar way to the FWER, as done in Section 2.2 and the Supporting Information Section \ref{SI:proofFWER}, as here we want the probability of rejecting any null hypotheses $H_{01},\hdots,H_{0K}$. Therefore if $\mu_k-\mu_0 = \theta_k$ for $k = 1, \hdots, K$, the event that $H_{01},\hdots,H_{0K}$ all fail to be rejected is equivalent to $\bar{R}_K(\Theta)$. The disjunctive power ($P_d$) for given $\Theta=(\theta_1,\hdots,\theta_K)$ is: 
\begin{equation}
P_d=1-P(\bar{R}_K(\Theta))=1-\sum_{\substack{j_k=1 \\ k=1,2,\ldots,K}}^{J_k} \Phi(\mathbf{L}^{+}_{\mathbf{j_k}}(\Theta),\mathbf{U}^{+}_{\mathbf{j_k}}(\Theta),\Sigma_{\mathbf{j_k}})
\label{Disjuctivepower}
\end{equation}
where $\mathbf{U}^+_{\mathbf{j_k}}(\Theta)=(U^+_{1,j_1}(\theta_1), \hdots, U^+_{K,j_K}(\theta_K))$ and $\mathbf{L}^+_{\mathbf{j_k}}(\Theta)=(L^+_{1,j_1}(\theta_1), \hdots, L^+_{K,j_K}(\theta_K))$ with $U^+_{k,j_k}(\theta_k)$ and $L^+_{k,j_k}(\theta_k)$ defined in Equation (6) and Equation (5)  of the main paper, respectively. The correlation matrix  $\Sigma_{\mathbf{j_k}}$ is the same as that given for FWER in Equation (3) of the main paper. 

\section{O'Brien and Fleming boundaries and the Pocock boundaries}
\label{SI:OBFandpo}
The O'Brien and Fleming boundaries \citep{OBrienPeterC.1979AMTP} with the futility boundaries equal to zero for $j<J_K$, to remove the symmetric boundaries, which may be too stringent \citep{MagirrD.2012AgDt} give the stopping boundaries
\begin{equation*}
 \begin{pmatrix}
U_1 \\
U_2  \\ 
\end{pmatrix}=
\begin{pmatrix}
3.166 & 2.239 \\
3.166 & 2.239  \\ 
\end{pmatrix}, 
\:\:\:\:
\begin{pmatrix}
L_1 \\
L_2 \\ 
\end{pmatrix}=
\begin{pmatrix}
0 & 2.239 \\
0 & 2.239 \\ 
\end{pmatrix}.
\end{equation*}  
When the focus is on ensuring that the pairwise power is greater than 80\% the sample sizes are
\begin{equation*}
\begin{pmatrix}
n_{1,1}  & n_{1,2} \\
 n_{2,1}  & n_{2,2}   
\end{pmatrix} =
\begin{pmatrix}
 70  & 140  \\
 70  & 140   
\end{pmatrix},
\:\:\:\:
\begin{pmatrix}
n_{0,1,1}  & n_{0,1,2} \\
 n_{0,2,1}  & n_{0,2,2}  
\end{pmatrix} =\begin{pmatrix}
 70  & 140 \\
 140  & 210
\end{pmatrix}.
\:\:\:\:
\begin{pmatrix}
n(1) \\
n(2)
\end{pmatrix} =\begin{pmatrix}
 0 \\
 70 
\end{pmatrix}.
\end{equation*}  
 When ensuring that the conjunctive power is greater than 80\% the sample sizes are
\begin{equation*}
\begin{pmatrix}
n_{1,1}  & n_{1,2} \\
 n_{2,1}  & n_{2,2}   
\end{pmatrix} =
\begin{pmatrix}
 87  & 174  \\
 87  & 174   
\end{pmatrix},
\:\:\:\:
\begin{pmatrix}
n_{0,1,1}  & n_{0,1,2} \\
 n_{0,2,1}  & n_{0,2,2}  
\end{pmatrix} =\begin{pmatrix}
 87  & 174 \\
 174  & 261
\end{pmatrix}.
\:\:\:\:
\begin{pmatrix}
n(1) \\
n(2)
\end{pmatrix} =\begin{pmatrix}
 0 \\
 87
\end{pmatrix}.
\end{equation*} 
Table \ref{tab:resultsobf} shows the results for different values of $\theta_1$ and $\theta_2$ when the conjunctive power is greater than 80\% and when the pairwise power is greater than 80\%.

\begin{table}[]
\centering
 \caption{Operating characteristics of the proposed designs under different values of $\theta_1$ and $\theta_2$, for both control of pairwise power and of conjunctive power, when the proposed designs use O'Brien and Fleming boundaries \citep{OBrienPeterC.1979AMTP} with futility boundaries equal to zero.}

\begin{tabular}{c|c|c|c|c|c|c|c}

\multicolumn{8}{c}{\textbf{Design for pairwise power}}
\\
\hline
$\theta_1$ & $\theta_2$ & $P_{PW,1}$ & $P_{PW,2}$ &  $P_{C}$   & $P_{D}$ & $\max(N)$ & $E(N|\theta_1,\theta_2)$  \\
\hline
 $\theta'$ & $\theta'$ & 0.806  & 0.806 & 0.671 & 0.941 & 490 & 452.3 \\
 $\theta'$ & $0$ & 0.806  & 0.013 & 0.806 & 0.807 & 490 & 407.3   \\
 $\theta'$ & $-\infty$ & 0.806  & 0 & 0.806 & 0.806 & 490 & 337.4  \\
  0 & $\theta'$ & 0.013  & 0.806 & 0.806 & 0.807 & 490 & 429.7   \\
  0 & $0$ & 0.013  & 0.013 & 1 & 0.025 & 490 & 384.8   \\
 $-\infty$ & $\theta'$ & 0 & 0.806 & 0.806 & 0.806 & 490 & 394.8   \\
\hline

\multicolumn{8}{c}{\textbf{Design for conjunctive power}}
\\
\hline
$\theta_1$ & $\theta_2$ & $P_{PW,1}$ & $P_{PW,2}$ &  $P_{C}$   & $P_{D}$ & $\max(N)$ & $E(N|\theta_1,\theta_2)$  \\
\hline
 $\theta'$ & $\theta'$ & 0.889  & 0.889 & 0.801 & 0.977 & 609 & 545.5   \\
 $\theta'$ & $0$ & 0.889  & 0.013 & 0.889 & 0.889 & 609 & 500.7   \\
 $\theta'$ & $-\infty$ & 0.889  & 0 & 0.889 & 0.889 & 609 & 413.8   \\
  0 & $\theta'$ & 0.013  & 0.889 & 0.889 & 0.889 & 609 & 523.1   \\
  0 & $0$ & 0.013  & 0.013 & 1 & 0.025 & 609 & 478.3   \\
 $-\infty$ & $\theta'$ & 0 & 0.889 & 0.889 & 0.889 & 609 & 479.6  \\
\end{tabular}
\label{tab:resultsobf} 
\end{table}	

\par

The Pocock boundaries \citep{PocockStuartJ.1977GSMi} with the futility boundaries equal to zero for $j<J_K$, to remove the symmetric boundaries, which may be too stringent \citep{MagirrD.2012AgDt} give the stopping boundaries
\begin{equation*}
 \begin{pmatrix}
U_1 \\
U_2  \\ 
\end{pmatrix}=
\begin{pmatrix}
2.440 & 2.440 \\
2.440 & 2.440  \\ 
\end{pmatrix}, 
\:\:\:\:
\begin{pmatrix}
L_1 \\
L_2 \\ 
\end{pmatrix}=
\begin{pmatrix}
0 & 2.440 \\
0 & 2.440 \\ 
\end{pmatrix}.
\end{equation*}  
When the focus is on ensuring that the pairwise power is greater than 80\% the sample sizes are: 
\begin{equation*}
\begin{pmatrix}
n_{1,1}  & n_{1,2} \\
 n_{2,1}  & n_{2,2}   
\end{pmatrix} =
\begin{pmatrix}
 76  & 152  \\
 76  & 152   
\end{pmatrix},
\:\:\:\:
\begin{pmatrix}
n_{0,1,1}  & n_{0,1,2} \\
 n_{0,2,1}  & n_{0,2,2}  
\end{pmatrix} =\begin{pmatrix}
 76  & 152 \\
 152  & 228
\end{pmatrix}.
\:\:\:\:
\begin{pmatrix}
n(1) \\
n(2)
\end{pmatrix} =\begin{pmatrix}
 76 \\
 152 
\end{pmatrix}.
\end{equation*}  
 When ensuring that the conjunctive power is greater than 80\% the sample sizes are: 
\begin{equation*}
\begin{pmatrix}
n_{1,1}  & n_{1,2} \\
 n_{2,1}  & n_{2,2}   
\end{pmatrix} =
\begin{pmatrix}
 95  & 190  \\
 95  & 190   
\end{pmatrix},
\:\:\:\:
\begin{pmatrix}
n_{0,1,1}  & n_{0,1,2} \\
 n_{0,2,1}  & n_{0,2,2}  
\end{pmatrix} =\begin{pmatrix}
 95  & 190 \\
 190  & 285
\end{pmatrix}.
\:\:\:\:
\begin{pmatrix}
n(1) \\
n(2)
\end{pmatrix} =\begin{pmatrix}
 0 \\
 95
\end{pmatrix}.
\end{equation*} 
Table \ref{tab:resultspo} shows the results for different values of $\theta_1$ and $\theta_2$ when the conjunctive power is greater than 80\% and when the pairwise power is greater than 80\%.

\begin{table}[]
\centering
 \caption{Operating characteristics of the proposed designs under different values of $\theta_1$ and $\theta_2$, for both control of pairwise power and of conjunctive power, when the proposed designs use Pocock boundaries \citep{PocockStuartJ.1977GSMi} with futility boundaries equal to zero. }

\begin{tabular}{c|c|c|c|c|c|c|c}

\multicolumn{8}{c}{\textbf{Design for pairwise power}}
\\
\hline
$\theta_1$ & $\theta_2$ & $P_{PW,1}$ & $P_{PW,2}$ &  $P_{C}$   & $P_{D}$ & $\max(N)$ & $E(N|\theta_1,\theta_2)$  \\
\hline
 $\theta'$ & $\theta'$ & 0.802  & 0.802 & 0.662 & 0.941 & 532 & 429.3 \\
 $\theta'$ & $0$ & 0.802  & 0.013 & 0.802 & 0.802 & 532  & 420.6   \\
 $\theta'$ & $-\infty$ & 0.802  & 0 & 0.802 & 0.802 & 532  & 345.7  \\
  0 & $\theta'$ & 0.013  & 0.802 & 0.802 & 0.803 & 532  & 424.9   \\
  0 & $0$ & 0.013  & 0.013 & 1 & 0.025 & 532  & 416.3   \\
 $-\infty$ & $\theta'$ & 0 & 0.802 & 0.802 & 0.802 & 532  & 387.5   \\
\hline

\multicolumn{8}{c}{\textbf{Design for conjunctive power}}
\\
\hline
$\theta_1$ & $\theta_2$ & $P_{PW,1}$ & $P_{PW,2}$ &  $P_{C}$   & $P_{D}$ & $\max(N)$ & $E(N|\theta_1,\theta_2)$  \\
\hline
 $\theta'$ & $\theta'$ & 0.889 & 0.889 & 0.801 & 0.978 & 665 & 507.6   \\
 $\theta'$ & $0$ & 0.889  & 0.013 & 0.889 & 0.890 & 665 & 516.1   \\
 $\theta'$ & $-\infty$ & 0.889  & 0 & 0.889 & 0.889 & 665 & 422.5   \\
  0 & $\theta'$ & 0.013  & 0.889 & 0.889 & 0.890 & 665 & 511.9   \\
  0 & $0$ & 0.013  & 0.013 & 1 & 0.025 & 665 & 520.4   \\
 $-\infty$ & $\theta'$ & 0 & 0.889 & 0.889 & 0.889 & 665 & 465.1  \\
\end{tabular}
\label{tab:resultspo} 
\end{table}	

\section{Non-binding triangular stopping boundaries}
\label{SI:nontri}
The triangular boundaries \citep{WhiteheadJ.1997TDaA} with non-binding futility boundaries for the type I error, are
\begin{equation*}
U = \begin{pmatrix}
U_1 \\
U_2  \\ 
\end{pmatrix}=
\begin{pmatrix}
2.517 & 2.373 \\
2.517 & 2.373 \\
\end{pmatrix}, 
\:\:\:\:
L =
\begin{pmatrix}
L_1 \\
L_2 \\ 
\end{pmatrix}=
\begin{pmatrix}
0.839 & 2.373 \\
0.839 & 2.373 \\
\end{pmatrix}.
\end{equation*}  
When the focus is on ensuring that the pairwise power is greater than 80\% the sample sizes are: 
\begin{equation*}
\begin{pmatrix}
n_{1,1}  & n_{1,2} \\
 n_{2,1}  & n_{2,2}   
\end{pmatrix} =
\begin{pmatrix}
 77  & 154  \\
 77  & 154   
\end{pmatrix},
\:\:\:\:
\begin{pmatrix}
n_{0,1,1}  & n_{0,1,2} \\
 n_{0,2,1}  & n_{0,2,2}  
\end{pmatrix} =\begin{pmatrix}
 77  & 154 \\
 154  & 231
\end{pmatrix}.
\:\:\:\:
\begin{pmatrix}
n(1) \\
n(2)
\end{pmatrix} =\begin{pmatrix}
 0 \\
 77 
\end{pmatrix}.
\end{equation*}  
 When ensuring that the conjunctive power is greater than 80\% the sample sizes are: 
\begin{equation*}
\begin{pmatrix}
n_{1,1}  & n_{1,2} \\
 n_{2,1}  & n_{2,2}   
\end{pmatrix} =
\begin{pmatrix}
 97  & 194  \\
97  & 194  \\  
\end{pmatrix},
\:\:\:\:
\begin{pmatrix}
n_{0,1,1}  & n_{0,1,2} \\
 n_{0,2,1}  & n_{0,2,2}  
\end{pmatrix} =\begin{pmatrix}
 97  & 194  \\
 194  & 291
\end{pmatrix}.
\:\:\:\:
\begin{pmatrix}
n(1) \\
n(2)
\end{pmatrix} =\begin{pmatrix}
 0 \\
 97
\end{pmatrix}.
\end{equation*} 
Table \ref{tab:resultsnontri} shows the results for different values of $\theta_1$ and $\theta_2$ when the conjunctive power is greater than 80\% and when the pairwise power is greater than 80\%. As can be seen in these results unlike in Table \ref{tab:resultstri} the disjunctive power no longer equals the target of 2.5\% when $\theta_1,\theta_2=0$. This is because this is the FWER if one did use the lower boundaries for futility. Without these lower bounds the FWER is 2.5\%. This is the same for the PWER when looking at the pairwise power when $\theta_1$ or $\theta_2$ equals 0.

\begin{table}[]
\centering
 \caption{Operating characteristics of the proposed designs under different values of $\theta_1$ and $\theta_2$, for both control of pairwise power and of conjunctive power, when the proposed designs use triangular boundaries \citep{WhiteheadJ.1997TDaA} with non-binding futility boundaries for the type I error.}

\begin{tabular}{c|c|c|c|c|c|c|c}

\multicolumn{8}{c}{\textbf{Design for pairwise power}}
\\
\hline
$\theta_1$ & $\theta_2$ & $P_{PW,1}$ & $P_{PW,2}$ &  $P_{C}$   & $P_{D}$ & $\max(N)$ & $E(N|\theta_1,\theta_2)$  \\
\hline
 $\theta'$ & $\theta'$ & 0.802  & 0.802 & 0.663 & 0.942 & 539 & 426.6 \\
 $\theta'$ & $0$ & 0.802  & 0.012 & 0.802 & 0.803 & 539 & 377.5   \\
 $\theta'$ & $-\infty$ & 0.802  & 0 & 0.802 & 0.802 & 539 & 347.5  \\
  0 & $\theta'$ & 0.012  & 0.802 & 0.802 & 0.804 & 539 & 402.0   \\
  0 & $0$ & 0.012 & 0.012 & 1 & 0.024 & 539 & 353.0   \\
 $-\infty$ & $\theta'$ & 0 & 0.802 & 0.802 & 0.802 & 539 & 387.0   \\
\hline

\multicolumn{8}{c}{\textbf{Design for conjunctive power}}
\\
\hline
$\theta_1$ & $\theta_2$ & $P_{PW,1}$ & $P_{PW,2}$ &  $P_{C}$   & $P_{D}$ & $\max(N)$ & $E(N|\theta_1,\theta_2)$  \\
\hline
 $\theta'$ & $\theta'$ & 0.891  & 0.891 & 0.803 & 0.979 & 679 & 513.9 \\
 $\theta'$ & $0$ & 0.891  & 0.012 & 0.891 & 0.891 & 679 & 467.7   \\
 $\theta'$ & $-\infty$ & 0.891  & 0 & 0.891 & 0.891 & 679 & 430.0   \\
  0 & $\theta'$ & 0.012  & 0.891 & 0.891 & 0.891 & 679 & 490.8   \\
  0 & $0$ & 0.012  & 0.012 & 1 & 0.024 & 679 & 444.7   \\
 $-\infty$ & $\theta'$ & 0 & 0.891 & 0.891 & 0.891 & 679 & 471.9  \\
\end{tabular}
\label{tab:resultsnontri} 
\end{table}	

\section{Plots based on the results from Section 4.3 for the two and three stage designs}
\label{SI:plotsS4.3}
The plots for the 2 stage and 3 stage example trials as given in Table 3 are shown in Figure \ref{fig:PW4plots} and Figure \ref{fig:TP4plots}. These plots are similar to the once seen in Figure 1 and 2 of the main paper. The y-axis gives the sample size for the trial. The x-axis gives the amount of control patients recruited between each active treatment being added ($n(k)-n(k-1)$). Plotted on the graph is  the maximum sample size and the expected sample size under the different configurations considered in Table \ref{tab:FWERvsPWER}. Figure \ref{fig:PW4plots} gives the plots when the pairwise power is controlled at 80\% and Figure \ref{fig:TP4plots} gives the plots when the conjunctive power is controlled at 80\%. As can be seen in Figure \ref{fig:TP4plots} there are times where there are less lines than expected. This is simply caused by the points when separate trials becomes better than running the proposed platform trial is at the same point for multiple different $\Theta$, as seen in Table \ref{tab:FWERvsPWER}, therefore the lines overlap. 

\begin{figure}[]

  \centering
  \includegraphics[width=.49\linewidth,trim= 0 0.5cm 0 0cm, clip]{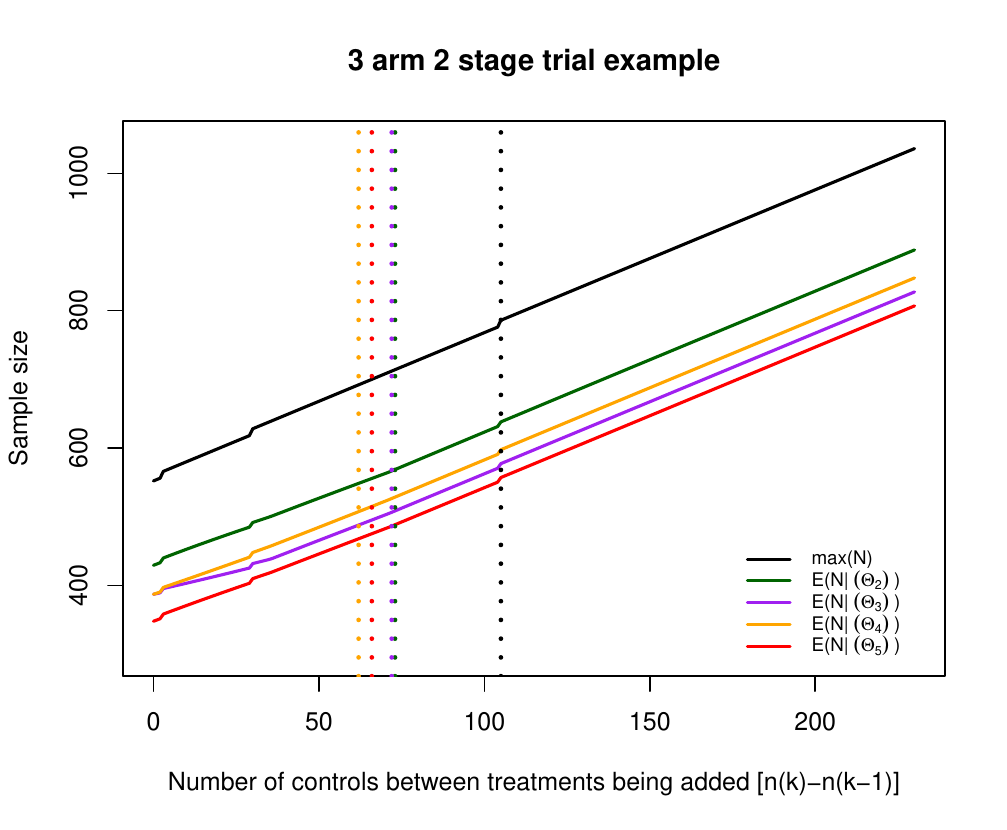}  
  \includegraphics[width=.49\linewidth,trim= 0 0.5cm 0 0cm, clip]{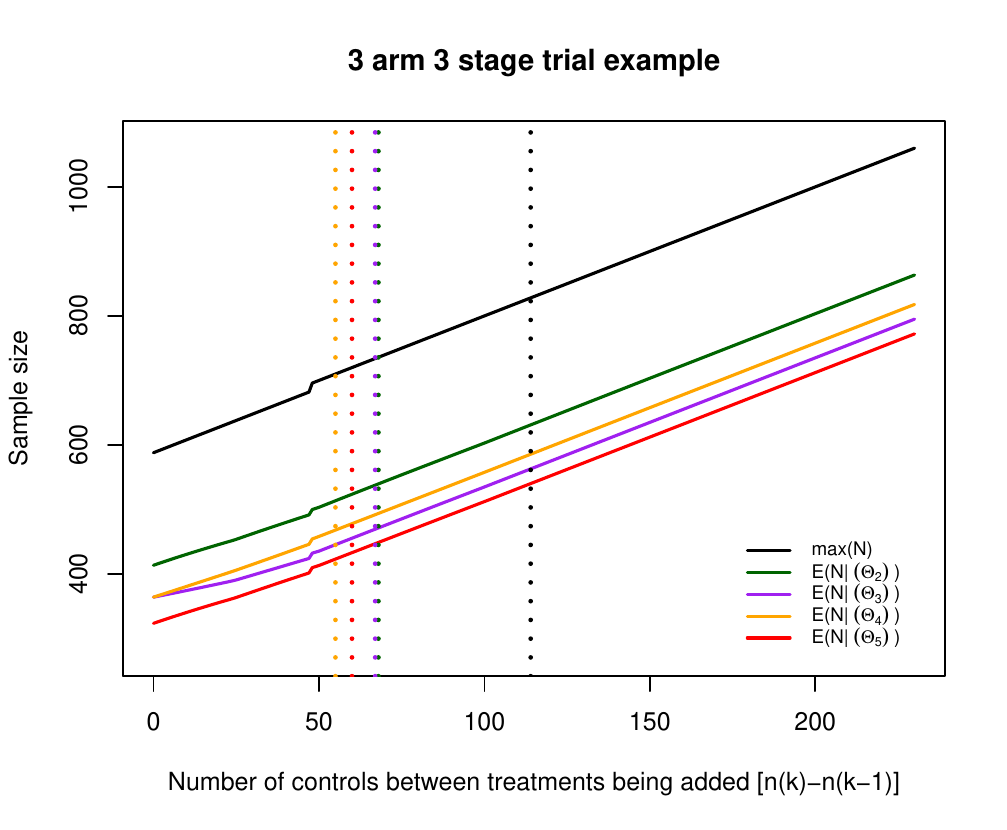}  
  \includegraphics[width=.49\linewidth,trim= 0 0.5cm 0 0cm, clip]{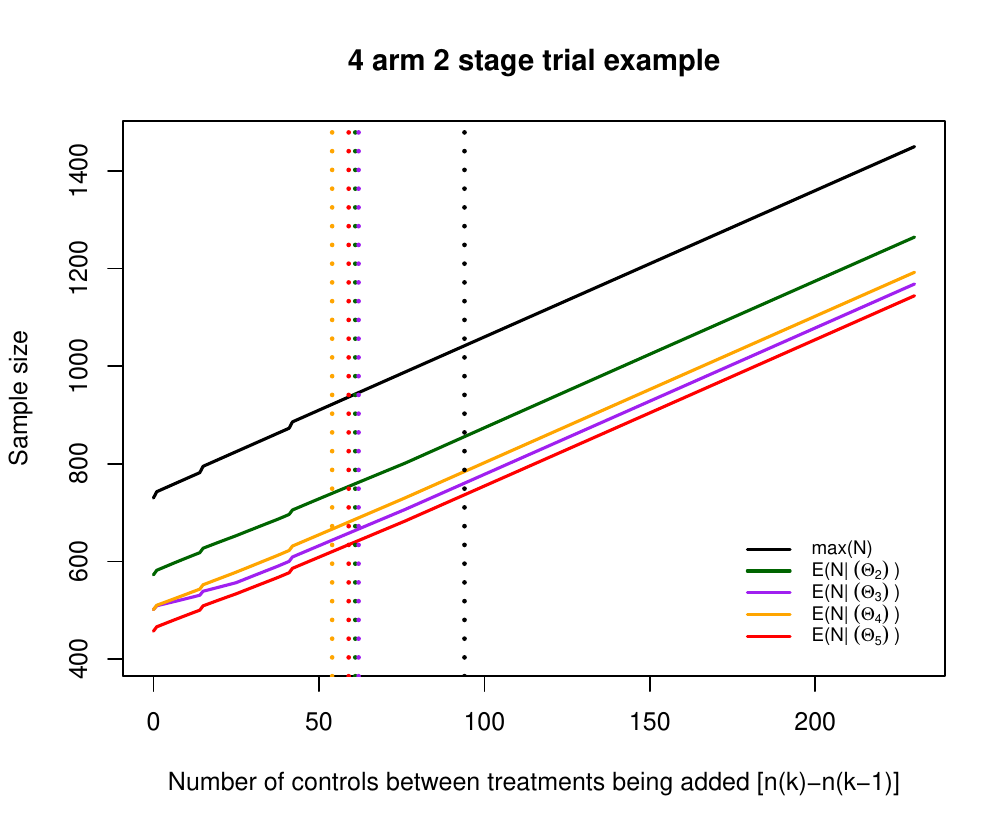}  
  \includegraphics[width=.49\linewidth,trim= 0 0.5cm 0 0cm, clip]{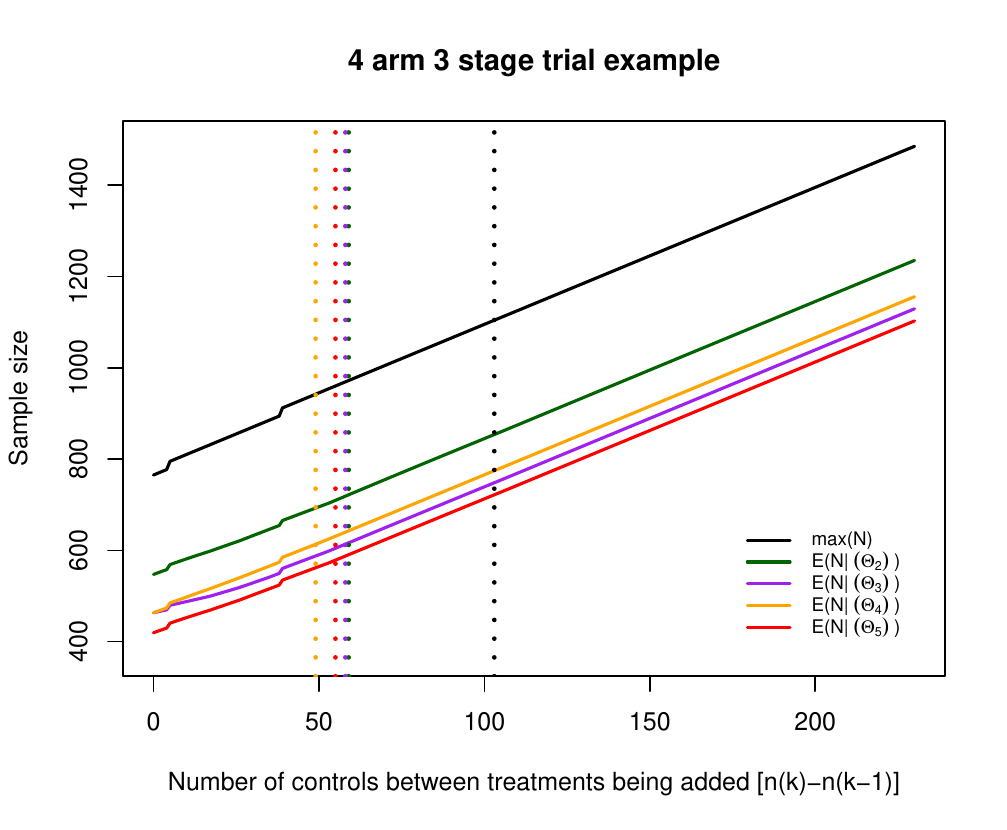}  

\caption{The maximum sample size and the expected sample size under different $\Theta$ depending on the value $n(k)-n(k-1)$, for the pairwise power control of 80\% and FWER of 5\% one-sided. The dash vertical lines correspond to the points where the maximum or expected sample size of the trial is now greater than running separate trials which each have type I error control of 2.5\% one-sided.}
\label{fig:PW4plots}
\end{figure}

\begin{figure}[]

  \centering
  \includegraphics[width=.49\linewidth,trim= 0 0.5cm 0 0cm, clip]{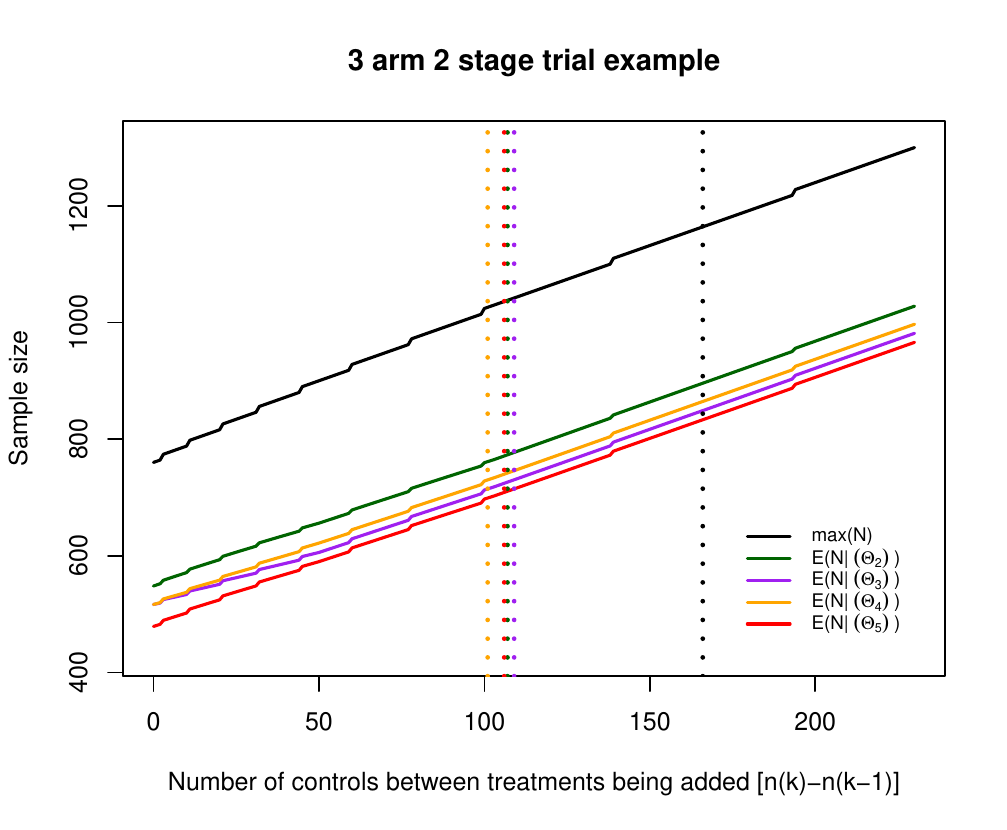}  
  \includegraphics[width=.49\linewidth,trim= 0 0.5cm 0 0cm, clip]{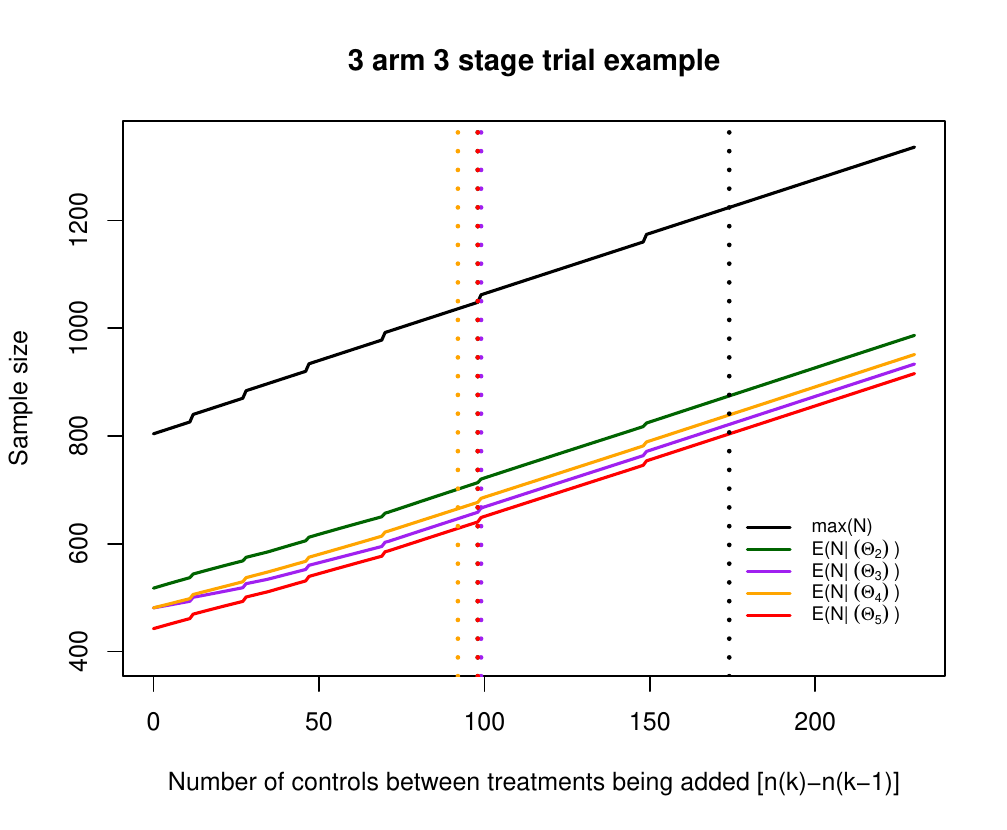}  
  \includegraphics[width=.49\linewidth,trim= 0 0.5cm 0 0cm, clip]{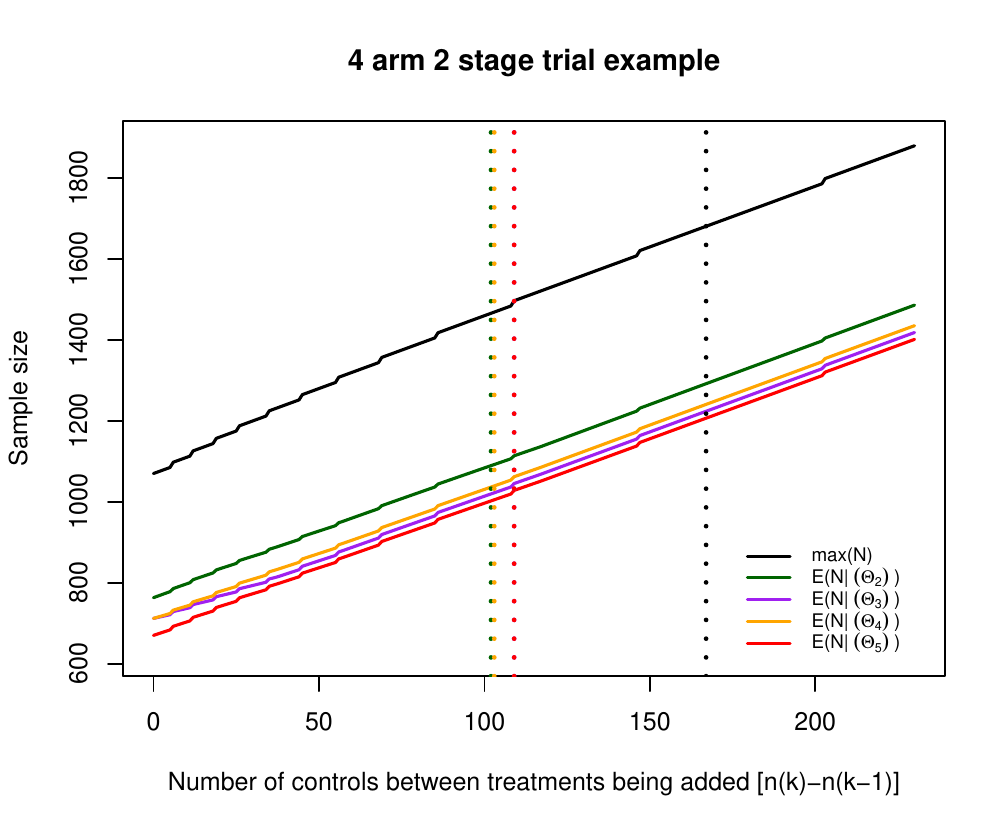}  
  \includegraphics[width=.49\linewidth,trim= 0 0.5cm 0 0cm, clip]{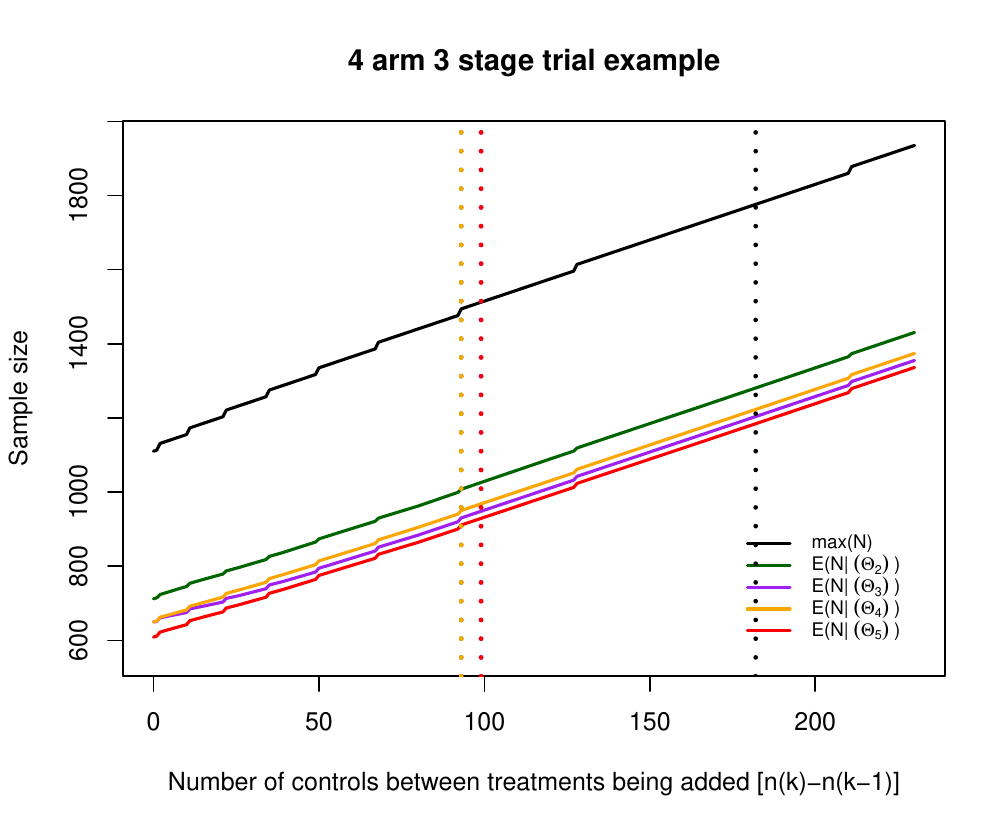}  

\caption{The maximum sample size and the expected sample size under different $\Theta$ depending on the value $n(k)-n(k-1)$, for the conjunctive power control of 80\% and FWER of 5\% one-sided. The dash vertical lines correspond to the points where the maximum or expected sample size of the trial is now greater than running separate trials which each have type I error control of 2.5\% one-sided.}
\label{fig:TP4plots}
\end{figure}

\end{document}